\documentclass[prd,nofootinbib,floats,aps,twocolumn,tightenlines,superscriptaddress]{revtex4-1}
\usepackage{natbib}
\usepackage{graphicx}
\usepackage{mathtools}
\usepackage{mathrsfs}
\usepackage{hyperref}
\usepackage{commath}
\usepackage{bm}
\usepackage{amsfonts,amsmath,amssymb,amsthm}
\usepackage{color}
\usepackage[usenames,dvipsnames]{xcolor}
\usepackage{subfigure}
\usepackage{dsfont}
\usepackage{enumitem}
\usepackage{qcircuit}

\DeclareMathOperator{\Tr}{Tr}
\DeclareMathOperator{\tr}{Tr}

\DeclareMathOperator{\sinc}{sinc}

\newcommand{\ket}[1]{| {#1} \rangle}
\newcommand{\bra}[1]{\langle {#1} |}
\newcommand{\ii}{\mathrm{i}}
\renewcommand{\a}[1]{\hat{a}_{\bm{#1}}}
\newcommand{\ad}[1]{\hat{a}_{\bm{#1}}^\dagger}
\newcommand{\Ua}{\hat U_\textsc{a}}
\newcommand{\Ub}{\hat U_\textsc{b}}
\newcommand{\U}[2]{\hat U_{\textsc{#1}_{#2}}}

\renewcommand*\d[2][]{%
	\mathrm{d}%
	\ifx\relax#1\relax\else
	\rule{-0.02em}{1.5ex}^{#1}\rule{0.08em}{0ex}\!
	\fi
	#2\,
}

\newcommand{\pia}{\hat\pi_\textsc{a}}
\newcommand{\phia}{\hat\phi_\textsc{a}}

\newcommand{\phih}{\hat \phi}

\newcommand{\pih}{\hat \pi}
\newcommand{\om}[1]{\omega_{\bm{#1}}}
\newcommand{\taa}{t_\textsc{a}}
\newcommand{\tbb}{t_\textsc{b}}
\newcommand{\OO}[1]{\hat O_{#1}}
\renewcommand{\P}[1]{\hat P_{#1}}

\newtheorem{mydef}{Definition}
\newtheorem{theorem}{Theorem}
\newtheorem{lemma}{Lemma}

\begin{document}

\title{Transmission of quantum information through quantum fields}
	
\author{Petar Simidzija}
\affiliation{Department of Applied Mathematics, University of Waterloo, Waterloo, Ontario, N2L 3G1, Canada}
\affiliation{Institute for Quantum Computing, University of Waterloo, Waterloo, Ontario, N2L 3G1, Canada}

\author{Aida Ahmadzadegan}
\affiliation{Department of Applied Mathematics, University of Waterloo, Waterloo, Ontario, N2L 3G1, Canada}
\affiliation{Perimeter Institute for Theoretical Physics, Waterloo, Ontario N2L 2Y5, Canada}

\author{Achim Kempf}
\affiliation{Department of Applied Mathematics, University of Waterloo, Waterloo, Ontario, N2L 3G1, Canada}
\affiliation{Institute for Quantum Computing, University of Waterloo, Waterloo, Ontario, N2L 3G1, Canada}
\affiliation{Perimeter Institute for Theoretical Physics, Waterloo, Ontario N2L 2Y5, Canada}
\affiliation{Department of Physics and Astronomy, University of Waterloo, Waterloo, Ontario, N2L 3G1, Canada}

\author{Eduardo Mart\'in-Mart\'inez}
\affiliation{Department of Applied Mathematics, University of Waterloo, Waterloo, Ontario, N2L 3G1, Canada}
\affiliation{Institute for Quantum Computing, University of Waterloo, Waterloo, Ontario, N2L 3G1, Canada}
\affiliation{Perimeter Institute for Theoretical Physics, Waterloo, Ontario N2L 2Y5, Canada}

\begin{abstract}
We investigate the quantum channel consisting of two localized quantum systems that communicate through a scalar quantum field. 
We choose a scalar field rather than a tensor or vector field, such as the electromagnetic field, in order to isolate the situation where the qubits are carried by the field amplitudes themselves rather than, for example, by encoding qubits in the polarization of photons. 
We find that suitable protocols for this type of quantum channel require the careful navigation of several constraints, such as the no-cloning principle, the strong Huygens principle and the tendency of short field-matter couplings to be entanglement breaking. We non-perturbatively construct a protocol for such a quantum channel that possesses maximal quantum capacity. 
\end{abstract}
	
\maketitle

\section{Introduction}
\label{sec:intro}

There exists much technology for sending information through the electromagnetic field but the underlying theory has traditionally been based on classical approximations of the senders, the fields, the receivers and of the transmitted information itself. Fundamentally, it is of course necessary to develop a theory for the transmission of information through fields that is fully quantized. This theory has been partially developed, in particular, with advances in the fields of quantum optics and relativistic quantum information theory. For instance, the rotating wave approximation has been used extensively in quantum optics to study the light-matter coupling~\cite{scully1999}, but it has recently been shown that in certain regimes such an approximation can lead to violations of causality by observers attempting to communicate through a quantum field~\cite{Compagno1989,Compagno1990,Clerk1998,Martinez2015,Funai2019}. It has also been studied how the Unruh-DeWitt (UDW) model~\cite{DeWitt1979} effectively describes the interactions of first quantized systems (such as atoms or qubits) with relativistic quantum fields, and it has been successfully used to explore the inherently relativistic aspects of coupling light to matter \cite{Pablo}. Recent progress in quantum technologies such as quantum cryptography, quantum computing and quantum sensing now lends urgency to the task of developing a theory that treats senders, fields, receivers and the transmitted information fully quantum theoretically.

Several new phenomena that arise from the quantum nature of senders, fields and receivers are already known. For example, it is known that when there are multiple emitters, then the overall radiation field that they emit can be shaped not only by suitably choosing relative phases among the emitters (as usual) but also by suitably pre-entangling the emitters  \cite{QSWAhmadzadegan}. It is also known that, in suitable circumstances, it is possible to transmit information through the field without transmitting energy to the receiver. Instead, there is an energetic expense for receiving the signal which is borne by the receiver. This result can be traced to the wave phenomenon \cite{McLenaghan1974} that the strong Huygens principle (which states that massless fields propagate only on but not in the light cone) can be violated, namely in $(1+1)$- and in $(1 + 2n)$-dimensional Minkowski spacetimes, as well as in all spacetimes of any dimensions if they possesses generic curvature. 

All of the above results about the quantum channel from a quantum sender via a quantum field to a quantum receiver concern the transmission of information that is classical. Concerning the transmission of actual quantum information (which involves the transmission of pre-established entanglement with an ancilla), it is possible and often convenient to encode qubits in spin degrees of freedom of the field, such as the polarization of photons, see e.g.,  \cite{Buttler1998,Hughes2002,Weier2006,Fedrizzi2009,Nauerth2013,Yin2017,Liao2018}, which can be used, for example, for the purposes of quantum key distribution~\cite{Bennett1984,Shor2000}.
However, when the quantum information is encoded in the polarization, interesting intricacies of the propagation of quantum information through a quantum field are bypassed, as the qubit simply `rides' on a carrier photon.  

Fundamentally, these intricacies are important, as quantum information can also be encoded in the field amplitudes themselves, irrespective of polarization, for example, even in a scalar field. In this context, it is known, for example, that if the interaction between the sender (or the receiver) and the quantum field is chosen too short then the interaction tends to become entanglement breaking \cite{Simidzija2017c}. It is also known that the no-cloning theorem \cite{Wootters1982} imposes strong restrictions on the quantum channel capacities among emitters and receivers, see \cite{Jonsson2018} for the case of $(1+1)$-dimensions. 

In the present paper, we will study the quantum channel capacity of the quantum channel consisting of an emitting quantum system that encodes the quantum information into a scalar quantum field in $(3+1)$-dimensions, the propagation and spread of the quantum information through the quantum field, and finally the receipt and decoding of the quantum information by an absorbing quantum system. We will thereby be able to study the fundamental constraints eluded to above, from the tendency for entanglement breaking in short interactions, to the no-cloning principle and the surprising implications of violations of the strong Huygens principle. We will demonstrate how these constraints can be navigated, at least in principle, by constructing a protocol which does achieve maximal quantum channel capacity. 

\section{Setup}\label{sec:setup}

In this section we will introduce background knowledge about the qubit-field interaction model we consider, and review the notion of quantum channel capacity.

\subsection{Unruh-DeWitt model}
\label{sec:UDW}

To study the ability of a pair of localized first-quantized quantized systems, Alice and Bob, such as atoms or molecules, to send and receive quantum information through a quantum field, we will consider a setup \cite{Cliche2010}
which models the quantum system of the sender (as well as thar of the receiver) as a single qubit by considering only two of its energy levels. In our setup, these qubits couple to a scalar field rather than a vector or tensor field because we are here interested not in encoding qubits in polarization degrees of freedom but in the field amplitudes themselves. The coupling between the sender (and receiver) system and the quantum field is modeled as a standard Unruh-DeWitt (UDW) interaction ~\cite{DeWitt1979}
which has been shown to produce qualitatively the same predictions as the full electromagnetic light-matter interaction in situations where the angular momentum exchange between light and matter can be ignored~\cite{Pozas2016,Pablo}. 
Crucially, in order to capture the new phenomena that we are interested in here, we will not make any of the simplifying assumptions that are often used in quantum optics, such as the single-mode or rotating wave approximation common in the Rabi, Glauber or Jaynes-Cummings light-matter interaction models \cite{scully1999}.

Concretely, we let Alice and Bob each locally couple their two-level quantum system (which we will refer to as an Unruh deWitt (UdW) detector) to a scalar quantum field $\hat\phi(\bm x,t)$. We take the free Hamiltonian of qubit $\nu\in\{\text{A},\text{B}\}$ to be $\hat H_\nu=\Omega_\nu\hat\sigma_z$, with $\Omega_\nu$ the energy gap and $\hat\sigma_z$ the Pauli z-operator. We denote the excited and ground states of $\hat H_\nu$ as $\ket{\pm_z}$, with eigenvalues $\pm \Omega_\nu$, respectively\footnote{Throughout this paper we will use the notation $\ket{\pm_s}$ for the eigenstates of $\hat\sigma_s$, $s\in\{x,y,z\}$, with eigenvalues of $\pm 1$. We will alternatively sometimes denote $\ket{+_z}$ as $\ket e$ and $\ket{-_z}$ as $\ket g$, when we want to emphasize that these are the ground and excited states of the free detector Hamiltonian.}. Meanwhile, the field $\hat\phi(\bm x,t)$, and its conjugate momentum $\hat\pi(\bm x,t)$, can be conveniently expanded in plane wave modes as 
\begin{align}
\label{eq:field}
	\hat\phi(\bm{x},t)
	&=
	\int\!\frac{\d[d]{\bm{k}}}{\sqrt{2(2\pi)^d |\bm k|}}\left(\ad{k} e^{\ii(|\bm k| t-\bm{k}\cdot\bm{x})}+\text{H.c.}\right)\!,
	\\
	\label{eq:pi_field}
	\hat\pi(\bm{x},t)
	&=
	\int\!\frac{\d[d]{\bm{k}}}{\sqrt{2(2\pi)^d |\bm k|}}
	\left(\ii|\bm k|\ad{k} e^{\ii(|\bm k| t-\bm{k}\cdot\bm{x})}+\text{H.c.}
	\right)\!,
\end{align}
where the creation and annihilation operators $\ad{k}$ and $\a{k}$ satisfy the canonical commutation relations
\begin{equation}
\label{eq:CCR}
	[\a{k},\a{k'}]= [\ad{k},\ad{k'}]=0, \quad
	[\a{k},\ad{k'}]=\delta^{(d)}(\bm{k}-\bm{k'}).
\end{equation}
Note that we are considering the field and the observers to live in a flat spacetime of arbitrary spatial dimension $d$. We will assume that before it interacts with the detectors, the field is in its ground state $\ket{0}$, defined by the condition $\a k\ket{0}=0$ for all momenta $\bm k\in \mathbb{R}^d$.

We can describe the interaction between comoving inertial detectors $\nu$ (where we use the label $\nu$ to denote Alice and Bob's detectors, $\nu\in\{\text{A},\text{B}\}$) and the field by specifying a local interaction Hamiltonian, $\hat H_{\textsc{i},\nu}(t)$. Working in the interaction picture of time evolution, we will consider interaction Hamiltonians of the form 
\begin{equation}
    \label{eq:H_I}
    \hat H_{\textsc{i},\nu}(t)
    =
    \lambda
    \chi(t)
    \hat m(t)
    \otimes 
    \hat O(t).
\end{equation}
Here $\lambda$ is a coupling strength, $\chi(t)$ is an explicitly time-dependent switching function, and $\hat m(t)$ and $\hat O(t)$ are qubit and field observables which contain an implicit time dependence coming from the fact that we are working in the interaction picture. For instance if the qubit couples through its $\hat\sigma_x$ observable, then $\hat m(t)$ is referred to as the monopole-moment operator, and reads
\begin{equation}
    \label{eq:m}
	\hat{m}(t)=
	\ket{+_z}\bra{-_z} e^{\ii\Omega_\nu t}+
	\ket{-_z}\bra{+_z} e^{-\ii\Omega_\nu t}.
\end{equation}
On the other hand, in order to ensure that the coupling between the observer $\nu$ and the field is physical, the field observable $\hat O(t)$ entering the interaction Hamiltonian must be local in spacetime (restricted to the region where the observer is located). To that end, for an observer coupling to the field at time $t$, we will allow $\hat  O(t)$ to be of the general form
\begin{equation}
\label{eq:Phi_local}
    \hat O(t)=
    \phih[F_1](t)+
    \pih[F_2](t),
\end{equation}
where for any field operator at a spacetime point $\hat{\mathcal{L}}(\bm x,t)$ we defined the smeared operator $\hat{\mathcal{L}}[F](t)$ as
\begin{align}
\label{eq:smeared_operator}
     \hat{\mathcal{L}}[F](t):=\int\d[d]{\bm x}F(\bm x)\hat{\mathcal{L}}(\bm x,t).
\end{align}
In order for $\hat O(t)$ in Eq.~\eqref{eq:Phi_local} to indeed be a local observable for the observer in question, the smearing functions $F_1$ and $F_2$ need to have support in the region of space at time $t$ where the observer is located. Note however that we do not require the observer to couple with the exact same smearing to the $\phih$ and $\pih$ fields; from a physical perspective this is analogous to an extended observer coupling his spatial profile differently to the electric and magnetic fields, as is actually the case in the light-matter interaction \cite{scully1999}. More complicated (i.e. non-linear) local field observables $\hat  O(t)$ could also be considered, although these can lead to divergences that must be carefully dealt with (see, e.g.~\cite{Hummer2016,Sachs2017}).

Following the specification of a qubit-field interaction Hamiltonian $\hat H_{\textsc{i},\nu}(t)$ as in Eq.~\eqref{eq:H_I}, we can formally write down the time-evolution unitary $\hat U$ generated by this Hamiltonian as
\begin{equation}
\label{eq:U}
	\hat{U}
	=
	\mathcal{T}\exp\left[{-\ii\int_{-\infty}^{\infty}\!\!\!\dif t\, \hat H_{\textsc{i},\nu}(t)}\right],
\end{equation}
where the $\mathcal T$ denotes the time-ordering operation. For general detector switching functions $\chi(t)$, the need for time-ordering makes a closed-form evaluation of time-evolved states impossible, instead allowing only for a perturbative approach to the problem. Of course, such a perturbative approach can only be taken when the coupling $\lambda$ between qubit and field is small with respect to the other scales of the problem. 

If instead a strong-coupling result is sought after, there are various non-perturbative methods that can be used (see, among others, \cite{Braun2002,Braun2005,Lin2007,Bruschi2013,Brown2013,Hotta2008,Hotta2009,Pozas2017,Simidzija2017c}). In this study, to attain non-perturbative results, we are going to consider the qubit detector switching function to be $\chi(t)=\sum_{i=1}^n\delta(t-t_i)$ with $t_i< t_{i+1}$, i.e. we require that the detector only interacts with the field at discrete instants in time. Then we can rewrite the time evolution unitary in Eq.~\eqref{eq:U} as $\hat U=\hat U_n\hat U_{n-1}\dots\hat U_1$, where $\hat U_i$ is defined as
\begin{equation}
    \hat U_i = \exp\left[-\ii \lambda\hat m(t_i)\otimes\hat  O(t_i)\right].
\end{equation}
The derivation of this result is shown explicitly in Appendix C of~\cite{Simidzija2017c}. Notice that the time-ordering operation $\mathcal T$ appearing in Eq.~\eqref{eq:U} has served its purpose by ensuring the unitaries $\hat U_i$ act in order of increasing time, and thereafter $\mathcal T$ no longer appears in the expression for $\hat U$. Therefore an exact analytical expression for the time evolved state of the detector-field system can be obtained, as will be exemplified later on. 

Although we will not make use of it in this paper, let us briefly mention here another related technique which allows for a non-perturbative study of relativistic light-matter interactions. Namely, as studied in detail in~\cite{Braun2002,Braun2005,Landulfo2016}, instead of avoiding the issues with time-ordering by requiring the qubit detectors to interact at discrete moments in time, one can alternatively avoid this issue by considering degenerate detectors. Namely, since degenerate detectors have free Hamiltonians which are proportional to the identity, their free time evolution is trivial, and hence this allows one to bypass the difficulties posed by the time ordering operation in a more indirect manner. More concretely, as discussed in \cite{Simidzija2018}, by performing a Magnus expansion~\cite{blanes_magnus_2009} of the evolution unitary $\hat U$ in Eq.~\eqref{eq:U} we find that in the case of a degenerate detector this expansion contains only two non-vanishing terms, thus allowing us to work non-perturbatively.

Now that we have an understanding of the Unruh-DeWitt light-matter interaction model, which we will use to describe the interactions of our observers Alice and Bob to a quantum field, let us now formulate more concretely the main problem of this paper.

Let us consider a tripartite quantum system composed of two qubits (A and B) and a massless scalar field $\hat\phi$. We assume that the field starts in its vacuum state $\ket{0}$ and that qubit B is initially in some predefined state $\hat\rho_{\textsc{b},0}$. Then, we define a quantum channel $\Xi$ from Alice to Bob as a map which takes as input a state $\hat\rho_{\textsc{a},0}$ on Alice's Hilbert space $\mathcal H_\textsc{a}$ and outputs a state $\Xi[\hat\rho_{\textsc{a},0}]$ on Bob's Hilbert space $\mathcal H_\textsc{b}$. Concretely, we write this channel as
\begin{equation}
\label{eq:Xi}
    \Xi[\hat\rho_{\textsc{a},0}]
    :=
    \Tr_{\textsc{a}\phi} 
    \left[
    \hat U_\textsc{b}
    \hat U_\textsc{a}
    \left(
    \hat\rho_{\textsc{a},0}
    \ket{0}\bra{0}
    \hat\rho_{\textsc{b},0}
    \right)
    \hat U_\textsc{a}^\dagger
    \hat U_\textsc{b}^\dagger
    \right],
\end{equation}
where $\hat U_\nu$ is a unitary between qubit $\nu$ and the field. Note that we are requiring qubit A to interact with the field prior to qubit B. A circuit diagram of the channel $\Xi$ is shown in Fig.~\ref{fig:Xi}

\begin{figure}
    \centering
    \includegraphics[width=0.9\linewidth]{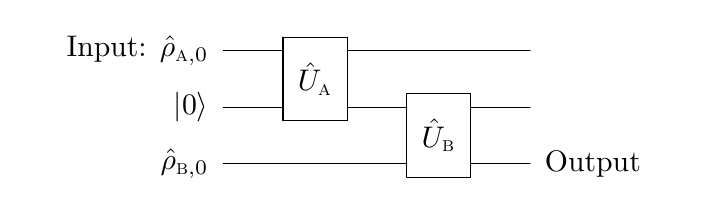}
    \caption[Quantum channel $\Xi$ from Alice to Bob via a quantum field]{Quantum channel $\Xi$ from Alice to Bob via a quantum field, which starts in its ground state $\ket 0$.}
    \label{fig:Xi}
\end{figure}

The problem which we are interested in can now be formulated as follows: Suppose that qubit A has access to the field in some region of spacetime centered at $(\bm x_\textsc{a},t_\textsc{a})$ and that qubit B couples to the field at some later time $t_\textsc{b}>t_\textsc{a}$. We ask two questions:
\begin{enumerate}[leftmargin=*]
    \item Can we construct local unitaries $\hat U_\textsc{a}$ and $\hat U_\textsc{b}$ such that the channel $\Xi$ is able to transmit quantum information from A to B?
    \item Is it possible for $\Xi$ to transmit quantum information perfectly? If so, where in space does qubit B have to be located? In other words: where does the quantum information that Alice puts into the field propagate?
\end{enumerate}
Answering these two questions is the main aim of this paper. However, before we can proceed with this, we must clarify what is meant by a channel being able to transmit quantum information. This is done in the following section.

\subsection{Quantum channel capacity and coherent information}
\label{sec:quant_channel_capacity}

Suppose that we have a channel $\Xi$ mapping the states of some Hilbert space $\mathcal H_\textsc{a}$ to the states of Hilbert space $\mathcal H_\textsc{b}$. In order to quantify the amount of quantum information that can be sent through the channel, we first need to define the \textit{coherent information} $I_c(\hat\rho_{\textsc{a},0},\Xi)$ associated with the channel $\Xi$ and an input to the channel $\hat\rho_{\textsc{a},0}$. To that end, we note that it is always possible to introduce a Hilbert space $\mathcal H_\textsc{c}$ and a state $\ket\psi \in \mathcal H_\textsc{c}\otimes \mathcal H_\textsc{a}$ such that $\hat \rho_{\textsc{a},0}=\Tr_\textsc{c}\ket\psi\bra\psi$; i.e. such that $\ket\psi$ is the \textit{purification} of $\hat\rho_{\textsc{a},0}$. Then, we set 
\begin{equation}
    \hat\rho_\textsc{cb}
    :=
    \left(\mathds 1_\textsc{c}
    \otimes
    \Xi\right)(\ket\psi\bra\psi),
\end{equation}
and we define the coherent information $I_c(\hat\rho_{\textsc{a},0},\Xi)$ as~\cite{Smith2010}
\begin{equation}
\label{eq:coherent_info}
    I_c(\hat\rho_{\textsc{a},0},\Xi)
    :=
    S(\hat\rho_\textsc{b})-
    S(\hat\rho_\textsc{cb}),
\end{equation}
where $\hat\rho_\textsc{b}:=\Xi(\hat\rho_{\textsc{a},0})$ and $S(\hat\rho):=-\Tr\hat\rho\log_2\hat\rho$ is the von Neumann entropy of the state $\hat\rho$, in units of bits. 

Although this definition of the coherent information is somewhat involved, it offers a very intuitive physical interpretation of its meaning. To see this, first note that before we put it through the channel $\Xi$, the system A was initially only entangled with the purifying system C. The coherent information $I_c(\hat\rho_{\textsc{a},0},\Xi)$ then quantifies how much of that entanglement between A and C is transferred to B and C. We can see this directly from Eq.~\eqref{eq:coherent_info}: the first term $S(\hat\rho_\textsc{b})$, being the entropy of system B, quantifies how entangled B is with the rest of the universe (i.e. with C as well as any additional ``channel environment" implicit in the definition of the channel $\Xi$), while the second term, $S(\hat\rho_\textsc{cb})$, is a measure of the entanglement of B and C with the channel environment. Hence the difference $I_c(\hat\rho_{\textsc{a},0},\Xi)=S(\hat\rho_\textsc{b})-S(\hat\rho_\textsc{cb})$ quantifies, in a way analogous to the classical mutual information~\cite{Cover2012},  the amount of correlations (in this case in the form of entanglement) between B and C. In particular, as we prove in Appendix~\ref{Appendix:quantum_info_results}, $I_c(\hat\rho_{\textsc{a},0},\Xi)>0$ only if $\hat\rho_\textsc{cb}$ is a non-separable state on $\mathcal H_\textsc{c}\otimes\mathcal H_\textsc{b}$.

With the above definition of the coherent information, we can now define the \textit{maximal coherent information} of the channel $\Xi$, denoted $I_\text{max}(\Xi)$, as
\begin{equation}
    I_\text{max}(\Xi):=
    \max_{\hat\rho_{\textsc{a},0}}
    I_c(\hat\rho_{\textsc{a},0},\Xi),
\end{equation}
where the maximization is taken over all inputs $\hat\rho_{\textsc{a},0}$ to the channel. Finally, the \textit{quantum channel capacity} $Q(\Xi)$ of the channel $\Xi$ can be defined as~\cite{Smith2010}
\begin{equation}
\label{eq:quantum_channel_capacity}
    Q(\Xi)=\lim_{n\rightarrow\infty}
    \frac{1}{n}I_\text{max}(\Xi^{\otimes n}).
\end{equation}
Physically, the quantum channel capacity $Q(\Xi)$ Can be understood as a sort of `average per instance of the channel' of the coherent information of $n$ copies of the channel. As such, it gives an idea of the amount of quantum information that can be coherently transmitted by the quantum channel, still taking into account that purposefully using $n$ copies of the channel can be better than $n$ independent uses ~\cite{Smith2010}. Unfortunately however, while the formula \eqref{eq:quantum_channel_capacity} provides an intuitive interpretation of the quantum capacity as being the maximal coherent information of many copies of the channel being allowed to work in parallel, it is generally not possible to evaluate this limit and obtain a closed form expression for $Q(\Xi)$. Instead, it is often only possible to compute lower bounds on $Q(\Xi)$, such as~\cite{Lloyd1997,Devetak2005}
\begin{equation}
    Q(\Xi)\ge
    I_\text{max}(\Xi)\ge
    I_c(\hat\rho_{\textsc{a},0},\Xi).
\end{equation}
In particular, computing $I_c(\hat\rho_{\textsc{a},0},\Xi)$ is generally much more straightforward than computing $Q(\Xi)$. For this reason, in what follows we will quantify the ability of the quantum channels proposed to transmit quantum information by computing its coherent information. Although this is only a lower bound on the full channel capacity $Q(\Xi)$, we will show that, in our case, we can construct the channel $\Xi$ so that it has a coherent information arbitrarily close to 1 with respect to the maximally mixed input state $\hat\rho_{\textsc{a},0}=\frac{1}{2}\mathds 1_\textsc{a}$. Since we will consider channels that take as input a single qubit state, we also know that necessarily $Q(\Xi)\le 1$. Hence we will show how to construct a channel transmitting qubits through a quantum field with a quantum channel capacity that is arbitrarily close to its maximal value. In other words, we will construct a perfect, field-mediated quantum channel.

\section{Constructing a perfect quantum channel}
\label{sec:perfect_quantum_channel}

Let us now proceed to construct a quantum channel $\Xi$ of the form~\eqref{eq:Xi}, which allows a (inertial) sender Alice to transmit a qubit of information through a quantum field to some future receiver Bob, who is relative rest with respect to Alice. To that end, let us suppose that Alice is located in a region of space characterized by the smearing function $F_\textsc{a}(\bm x)$, and that she wishes to encode her message into the field at time $t_\textsc{a}$. Then we would like to answer the following questions: where in space should Bob be located at time $t_\textsc{b}>t_\textsc{a}$, and what should the unitaries $\hat U_\textsc{a}$ and $\hat U_\textsc{b}$ be, in order for Bob to recover the entirety of Alice's message?

Before we proceed to answering these questions, let us remind ourselves that our ultimate goal is not just to construct a free space quantum channel between two observers --- indeed, as discussed in the introduction, free space quantum channels have already been realized experimentally over distances as large as 1000+ km, typically in the context of establishing entangled pairs between distant receivers for the purposes of implementing quantum key distribution protocols~\cite{Buttler1998,Hughes2002,Weier2006,Fedrizzi2009,Nauerth2013,Yin2017,Liao2018}. Rather, our goal is to understand, from a fundamental relativistic quantum information perspective, exactly how quantum information is encoded into, propagated through, and decoded out of, a quantum field, and where this quantum information is localized in spacetime. To that end, let us attempt to construct the channel $\Xi$ to be as simple as possible, so that we may try to understand its essential features without being distracted by the unessential ones.

With this additional requirement of simplicity for our channel $\Xi$, let us attempt to generate the time-evolution unitaries $\hat U_\textsc{a}$ and $\hat U_\textsc{b}$ defining the channel out of the simplest possible type of interaction Hamiltonians: those which couple the qubits A and B to the field only at discrete instants in time. Conveniently, as discussed in Sec.~\ref{sec:UDW}, by constructing our channel out of these simple couplings, we have the additional advantage that our analysis will be fully non-perturbative.

In fact, the ability to study our problem non-perturbatively is not merely a nice convenience that arises out of using simple-generated interaction unitaries. More crucially, as we will now prove, if we want the quantum channel $\Xi$ from Alice to Bob to be a perfect quantum channel (i.e. to have a maximum possible quantum channel capacity $Q(\Xi)=1$), then it is \textit{necessary} that the channel is constructed out of non-perturbative couplings between Alice and Bob's qubits and the field.

The proof of this claim is rather trivial. Let us suppose that the coupling between the qubits and field is quantified by some coupling strength $\lambda$, and let us consider the quantum channel capacity $Q(\Xi)$ as a power series in $\lambda$. Clearly if $\lambda=0$, the qubits A and B do not couple to the field, and hence we would have $Q(\Xi)=0$. Hence we can write
\begin{equation}
    Q(\Xi)=0+\mathcal O(\lambda),
\end{equation}
and thus we see that if the coupling $\lambda$ is weak --- in units set by the other scales in the problem --- then $Q(\Xi)$ would, at best, only differ by a small amount (i.e. an amount much less than one) from zero. Hence if we want to have $Q(\Xi)\approx 1$, we must consider couplings $\lambda$ in the non-perturbative regime.

\subsection{Simple-generated couplings}

The simplest possible unitaries $\hat U_\textsc{a}$ and $\hat U_\textsc{b}$ coupling the qubits A and B to the field are of the form
\begin{equation}
\label{eq:U_rank1}
    \hat U_\nu = \exp\left(\ii\lambda_\nu \hat m_\nu\otimes\hat O_\nu\right),
\end{equation}
where we abbreviate $\hat m_\nu:=\hat m(t_\nu)$ for the qubit observables and $\hat O_\nu:=\hat O(t_\nu)$ for the field observables. We will call unitaries of this form ``simple-generated" or ``rank-1 unitaries" because they are the exponential of a simple Schmidt rank-1 tensor product of qubit-field observables. We claim that if either $\hat U_\textsc{a}$ or $\hat U_\textsc{b}$ are of the simple-generated form, then the channel $\Xi$ in Eq.~\eqref{eq:Xi} is not able to transmit quantum information.

To prove this claim, let us first use our requirement that qubit A interacts with the field before qubit B to decompose the channel $\Xi$ from A to B in Eq.~\eqref{eq:Xi} to read $\Xi=\Xi_\textsc{b}\circ\Xi_\textsc{a}$, where $\Xi_\textsc{a}$ is a channel from A to the field defined by
\begin{equation}
    \Xi_\textsc{a}(\hat\rho_\textsc{a})
    :=
    \Tr_\textsc{a}
    \left[
    \Ua
    \left(\hat\rho_\textsc{a}\otimes\ket 0\bra 0
    \right)
    \Ua^\dagger
    \right],
\end{equation}
while $\Xi_\textsc{b}$ is a channel from the field to B defined by
\begin{equation}
    \Xi_\textsc{b}(\hat\rho_\phi)
    :=
    \Tr_\phi
    \left[
    \Ub
    \left(\hat\rho_\phi\otimes\hat\rho_{\textsc{b},0}
    \right)
    \Ub^\dagger
    \right].
\end{equation}
Next, we note that in Ref.~\cite{Simidzija2018} it was shown that unitaries of the simple-generated form necessarily give rise to \textit{entanglement breaking channels}, which are defined as follows:

\begin{mydef}\label{def:ent_breaking_channel}
A channel $\chi$ from states on $\mathcal H_\textsc{a}$ to states on $\mathcal H_\textsc{b}$ is said to be \emph{entanglement breaking} if for any state $\hat\rho_\textsc{ac}$ on $\mathcal H_\textsc{a}\otimes\mathcal H_\textsc{c}$ the state $(\chi\otimes\mathds{1}_\textsc{c})(\hat\rho_\textsc{ac})$ on $\mathcal H_\textsc{b}\otimes\mathcal H_\textsc{c}$ is separable.
\end{mydef}

Thus we find that if $\hat U_\nu$ is simply-generated, then $\Xi_\nu$ is an entanglement breaking channel. However, we previously noted that the coherent information in Eq.~\eqref{eq:coherent_info} is greater than zero only if entanglement can be transferred through the channel. Hence, from the above definition of an entanglement breaking channel, we find that if either $\hat U_\textsc{a}$ or $\hat U_\textsc{b}$ are of the simple-generated form then the maximal coherent information $I_c(\Xi)$ of the channel $\Xi$ is zero.

Notice however, that this does not yet prove that the quantum capacity $Q(\Xi)$ of the channel $\Xi$ is itself zero, since $I_c(\Xi)$ is only a lower bound on $Q(\Xi)$. Let us now show that it is indeed the case that $Q(\Xi)=0$ if either $\Ua$ or $\Ub$ are simply-generated unitaries.

The proof of this stronger claim starts with the expression Eq.~\eqref{eq:quantum_channel_capacity} for the quantum channel capacity of $\Xi$ in terms of the maximal coherent information of $\Xi^{\otimes n}$ for large $n$. Using the fact that $\Xi^{\otimes n} = \Xi_\textsc{b}^{\otimes n}\circ\Xi_\textsc{a}^{\otimes n}$ we note that if either of the $\Xi_\nu^{\otimes n}$ are entanglement breaking then $\Xi^{\otimes n}$ will have a maximal coherent information of zero (as discussed above), and hence $Q(\Xi)=0$. To that end, let us start by proving that $\Xi_\textsc{a}^{\otimes n}$ is entanglement breaking.

To prove that $\Xi_\textsc{a}^{\otimes n}$ is entanglement breaking, we need to show that for any Hilbert space $\mathcal H_\textsc{c}$, and any (potentially entangled) state $\hat \rho_{\textsc{c},\textsc{a}^n}\in\mathcal H_\textsc{c}\otimes\mathcal H_\textsc{a}^{\otimes n}$, the state $(\mathds{1}_\textsc{c}\otimes\Xi_\textsc{a}^{\otimes n})(\hat \rho_{\textsc{c},\textsc{a}^n})$ is a separable state on $\mathcal H_\textsc{c}\otimes\mathcal H_\phi^{\otimes n}$. We thus compute
\begin{align}
\label{eq:ent_breaking_1}
    &(\mathds{1}_\textsc{c}\otimes\Xi_\textsc{a}^{\otimes n})(\hat \rho_{\textsc{c},\textsc{a}^n})
    \\
    =&
    \notag
    \Tr_{\textsc{a}^n}
    \Big[
    \U a 1\dots\U a n
    \left(
    \hat \rho_{\textsc{c},\textsc{a}^n}
    \otimes
    \left(
    \ket 0 \bra 0
    \right)^{\otimes n}
    \right)
    \U a n ^\dagger\dots \U a 1 ^\dagger
    \Big]\!.
\end{align}
To proceed we note that we can express any qubit observable $\hat m_\textsc{a}$ as
\begin{equation}
    \hat m_\textsc{a}
    =
    \sum_{s\in\{\pm\}}
    s \hat P_s,
\end{equation}
where $\hat P_{\pm}:=\ket{\pm_z}\bra{\pm_z}$ are projectors onto the $\pm$ eigenstates of $\hat\sigma_z$. With this decomposition of $\hat m_\textsc{a}$, we can write $\hat U_\textsc{a}$ from Eq.~\eqref{eq:U_rank1} as 
\begin{equation}
\label{eq:controlled_unitary}
    \hat{U}_\textsc{a} 
    = 
    \sum_{s\in\{\pm\}}
    \hat{P}_s\otimes\hat{U}_s,
\end{equation}
where $\hat U_s:=\exp(\ii s \lambda\hat O_\textsc{a})$. Note that when written in this form, it is manifest that a rank-1 unitary $\hat U_\textsc{a}$ can be viewed a controlled unitary, where the state $\ket{s_z}$ of the qubit controls the unitary operation $\hat U_s$ performed on the field. Inserting Eq.~\eqref{eq:controlled_unitary} into Eq.~\eqref{eq:ent_breaking_1} we obtain
\begin{align}
\label{eq:ent_breaking_2}
    &(\mathds{1}_\textsc{c}\otimes\Xi_\textsc{a}^{\otimes n})(\hat \rho_{\textsc{c},\textsc{a}^n})
    \\
    =&
    \sum_{s_i,s'_i}
    \notag
    \Tr_{\textsc{a}^n}
    \Big[
    \hat P_{s_1}\dots\hat P_{s_n}
    \Big(
    \hat \rho_{\textsc{c},\textsc{a}^n}
    \otimes
    \hat U_{s_1}
    \ket 0 \bra 0
    \hat U_{s'_1}^\dagger
    \otimes
    \dots
    \notag\\
    &\hspace{5em}
    \otimes
    \hat U_{s_n}
    \ket 0 \bra 0
    \hat U_{s'_n}^\dagger
    \Big)
    \hat P_{s'_n}\dots\hat P_{s'_1}
    \Big],
\end{align}
where the sub-index $i$ runs from 1 to $n$. Using the cyclicity of the partial trace with respect to the system being traced over, and the fact that $\hat P_{s_i}\hat P_{s'_i} = \hat P_{s_i}\delta_{s_i s'_i}$ since $\hat P_{s_i}$ are projectors, Eq.~\eqref{eq:ent_breaking_2} simplifies to
\begin{align}
\label{eq:ent_breaking_3}
    \nonumber &(\mathds{1}_\textsc{c}\otimes\Xi_\textsc{a}^{\otimes n})(\hat \rho_{\textsc{c},\textsc{a}^n})
    \\
    =&
    \sum_{s_i}
    \notag
    \Tr_{\textsc{a}^n}
    \Big[
    \hat P_{s_1}\dots\hat P_{s_n}
    \Big(
    \hat \rho_{\textsc{c},\textsc{a}^n}
    \otimes
    \hat U_{s_1}
    \ket 0 \bra 0
    \hat U_{s_1}^\dagger
    \otimes
    \dots
    \notag\\
    &\hspace{5em}
    \otimes
    \hat U_{s_n}
    \ket 0 \bra 0
    \hat U_{s_n}^\dagger
    \Big)
    \Big].
\end{align}
Finally, defining
\begin{align}
    p(s_1,\dots,s_n)
    &:=
    \Tr
    \left[
    \hat P_{s_1}\dots\hat P_{s_n}
    \hat \rho_{\textsc{c},\textsc{a}^n}
    \right],
    \\
    \hat \rho_{\textsc{c}}(s_1,\dots,s_n)
    &:=
    \frac{1}{p(s_1,\dots,s_n)}
    \Tr_{\textsc{a}^n}
    \left[
    \hat P_{s_1}\dots\hat P_{s_n}
    \hat \rho_{\textsc{c},\textsc{a}^n}
    \right],
    \\
    \hat \rho_{\phi^n}(s_1,\dots,s_n)
    &:=
    \hat U_{s_1}
    \ket 0 \bra 0
    \hat U_{s_1}^\dagger
    \dots
    \hat U_{s_n}
    \ket 0 \bra 0
    \hat U_{s_n}^\dagger,
\end{align}
where $\hat \rho_{\textsc{c}}(s_1,\dots,s_n)$ is a density matrix on the Hilbert space $\mathcal H_\textsc{c}$, $\hat \rho_{\phi^n}(s_1,\dots,s_n)$ is a density matrix on $\mathcal H_\phi^{\otimes n}$, and $p(s_1,\dots,s_n)\ge 0$ with $\sum_{s_i}p(s_1,\dots,s_n)=1$, we can write $(\mathds{1}_\textsc{c}\otimes\Xi_\textsc{a}^{\otimes n})(\hat \rho_{\textsc{c},\textsc{a}^n})$ as
\begin{equation}
    \sum_{s_i}
    p(s_1,\dots,s_n)\,
    \hat \rho_{\textsc{c}}(s_1,\dots,s_n)
    \otimes
    \hat \rho_{\phi^n}(s_1,\dots,s_n),
\end{equation}
which is a manifestly separable state on $\mathcal H_\textsc{c}\otimes\mathcal H_\phi^{\otimes n}$. Hence the n-qubit channel $\Xi_\textsc{a}^{\otimes n}$ is entanglement breaking for all integers $n>0$.

In an analogous fashion, we can prove that the channel $\Xi_\textsc{b}^{\otimes n}$ from the field to B is entanglement breaking, for all integers $n>0$. The only subtlety with this proof compared to the one we just presented for $\Xi_\textsc{a}^{\otimes n}$, is that, because we are now considering a channel from the field to a qubit, rather than vice-versa, we have to perform a spectral decomposition of the field observable, rather than the qubit observable. To perform this decomposition rigorously is not trivial since the field observables acts on an uncountably infinite dimensional Hilbert space. Nevertheless, since the field observables are self-adjoint operators it is possible to apply the spectral theorem to them, and hence obtain such a spectral decomposition (see Appendix A of Ref.~\cite{Simidzija2018} for full details of this calculation). In this manner, we can show that a rank-1 unitary between a qubit and a field can not only be written as a controlled unitary from the qubit to the field, as in Eq.~\eqref{eq:controlled_unitary}, but also as a controlled unitary from the field to the qubit. In this way, the same argument that was used above to show that $\Xi_\textsc{a}^{\otimes n}$ is an entanglement breaking channel can also be used to arrive at the same conclusion for $\Xi_\textsc{b}^{\otimes n}$.

In conclusion we have shown that if either of the unitaries $\Ua$ or $\Ub$ used to define the channel $\Xi$ are of the simple-generated form, then the quantum channel capacity $Q(\Xi)$ is necessarily zero. In other words, simple-generated couplings between Alice and Bob's qubits to the field are too simple for the purposes of transmitting quantum information through the field. Thus in order to achieve quantum information transmission, we will need to consider more complicated couplings.

\subsection{Encoding a qubit into a field}

In our attempt to construct a channel $\Xi$ that allows a spacetime emitter A to send a qubit through a quantum field to a receiver B, we have come to the important conclusion that such a channel is not possible if either of the observers couple to the field through simple-generated unitaries $\hat U_\nu = \exp(\ii\lambda_\nu \hat m_\nu\otimes\hat O_\nu)$. The natural way to proceed with constructing $\Xi$ is to consider the next simplest types of interaction unitaries $\hat U_\nu$, composed of two rank-1 unitaries performed one after the other, i.e.
\begin{equation}
\label{eq:double_delta_unitary}
    \hat U_\nu = 
    \exp\left(\ii\lambda_{\nu2} \hat m_{\nu2}\otimes\hat O_{\nu2}\right)
    \exp\left(\ii\lambda_{\nu1} \hat m_{\nu2}\otimes\hat O_{\nu1}\right),
\end{equation}
where the $\lambda_{\nu i}$ are coupling constants, the $\hat m_{\nu i}$ are qubit observables and the $\hat  O_{\nu i}$ are field observables. As we will now show, we can indeed find unitaries of this form which ensure that the quantum capacity of the channel $\Xi$ is not only non-zero, but is in fact arbitrarily close to its theoretically maximal value of 1. 

To understand this construction of the unitaries $\Ua$ and $\Ub$, it will be instructive to first consider a simple example of a quantum channel that we know has perfect quantum capacity. To that end, let us consider the setup in which Alice and Bob would like to transmit a qubit of information by encoding and decoding their message into and out of a third qubit, F, rather than into and out of the quantum field. In this case, we know that the channel shown in Fig.~\ref{fig:swap_channel}, which simply swaps qubit A with F, and then F with B, is clearly able to perfectly transfer qubit A to qubit B.

\begin{figure}
    \begin{subfigure}
    \centering
    \includegraphics[width=0.9\linewidth]{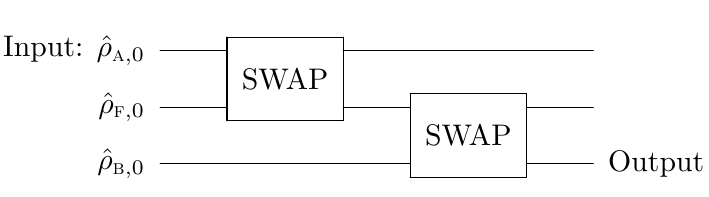}
    \caption[Perfect quantum channel from Alice to Bob via a third qubit, F]{Perfect quantum channel from Alice to Bob via a third qubit, F.}
    \label{fig:swap_channel}
    \end{subfigure}
    \begin{subfigure}
    \centering
   \includegraphics[width=0.9\linewidth]{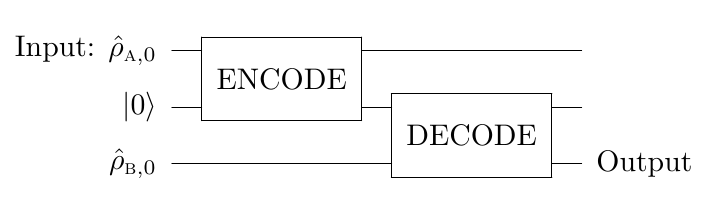}
    \caption[Perfect quantum channel from Alice to Bob via a quantum field, $\hat \phi$]{Perfect quantum channel from Alice to Bob via a quantum field, $\hat \phi$.}
    \label{fig:encode_channel}
    \end{subfigure}
\end{figure}

The key question thus becomes: Is it possible to construct a channel analogous to the one in Fig.~\ref{fig:swap_channel} if we take the intermediary system to be a quantum field $\hat\phi$ rather than a single qubit F? Indeed, our expectation is that it should be possible, if only from the perspective that we could exclusively couple the qubits to a two-dimensional subspace of the field's Hilbert space, which is effectively equivalent to coupling to a third qubit F. However there is one important distinction between this field mediated channel and the qubit mediated channel in Fig.~\ref{fig:swap_channel}. Namely, while in the latter case it makes sense to construct the channel out of SWAP gates, it does not make sense to talk about a SWAP gate between a qubit and a field, since their Hilbert spaces have different dimensions and are thus not isomorphic as vector spaces. In other words, it is not possible to have a one-to-one identification between the two basis states which span the qubit's Hilbert space with the infinitely many basis states spanning the field's Hilbert space --- any such identification would necessarily lose information about the field degrees of freedom. To belabour this point we will call the gate which encodes a qubit into a field an ENCODE gate, rather than a SWAP gate, and similarly for the DECODE gate. Indeed it should be possible to encode (and decode) the two-dimensional Hilbert space of the qubit in the infinite dimensional Hilbert space of the field. Thus the field-mediated quantum channel which we are trying to construct is shown in Fig.~\ref{fig:encode_channel}.

Our question therefore becomes the following: How can an observer Alice, coupling to the field at time $t_\textsc{a}$ and with a spatial extent given by the support of the smearing function $F_\textsc{a}(\bm x)$, encode the state of her qubit into the field? 

Let us suppose that Alice's qubit is in some arbitrary pure state $c_1\ket{+_z}+c_2\ket{-_z}$ and that the field is initially in the vacuum state $\ket 0$. Consider the qubit-field unitary $\Ua$ given by
\begin{equation}
\label{eq:Ua}
    \Ua
    =
    \exp\big(\ii\hat\sigma_x\pia\big)
    \exp\big(\ii\hat\sigma_z\phia\big),
\end{equation}
where $\phia$ and $\pia$ are smeared field observables defined as
\begin{align}
\label{eq:phia}
    \phia&:=\lambda_\phi
    \int\d[d]{\bm x}F_\textsc{a}(\bm x)
    \hat\phi(\bm x,t_\textsc{a}),
    \\
\label{eq:pia}
    \pia&:=\lambda_\pi
    \int\d[d]{\bm x}F_\textsc{a}(\bm x)
    \hat\pi(\bm x,t_\textsc{a}).
\end{align}
Note that the coupling constants have dimensions of $[\lambda_\phi]=L^\frac{d-1}{2}$ and $[\lambda_\pi]=L^\frac{d+1}{2}$ where $d$ is the number of spatial dimensions of spacetime. Also note that the unitary $\Ua$ is generated by the interaction Hamiltonian
\begin{align}
    \hat H_\textsc{i,a}(t)
    =\,
    &\lambda_\phi
    \delta(t-t_\textsc{a}^-)
    \hat m_\textsc{a}^z(t)
    \otimes
    \int\d[d]{\bm x}F_\textsc{a}(\bm x)
    \hat\phi(\bm x,t)\\
    &+
    \lambda_\pi
    \delta(t-t_\textsc{a}^+)
    \hat m_\textsc{a}^x(t)
    \otimes
    \int\d[d]{\bm x}F_\textsc{a}(\bm x)
    \hat\pi(\bm x,t),\notag
\end{align}
where $\hat m_\textsc{a}^z(t)$ is the $\hat\sigma_z$ operator in the interaction picture (so $\hat m_\textsc{a}^z(t)=\hat\sigma_z$ for all $t$ since $\hat\sigma_z$ is proportional to the detector's free Hamiltonian), $\hat m_\textsc{a}^x(t)$ is the $\hat\sigma_x$ operator in the interaction picture (i.e. it is the monopole moment operator from Eq.~\eqref{eq:m}), and the times $t_\textsc{a}^\pm\approx t_\textsc{a}$ are such that $t_\textsc{a}^-$ is just slightly less than $t_\textsc{a}^+$.\footnote{Physically this amounts to saying that $t_\textsc{a}^+=t_\textsc{a}^-+\epsilon$ where $0<\epsilon\ll\Omega^{-1}$, with $\Omega$ being the free frequency of the detector.}

We now claim that the unitary $\Ua$ in Eq.~\eqref{eq:Ua} effectively encodes the state of the qubit in the state of the field, as long as the following two conditions are satisfied:
\begin{align}
\label{eq:strong_coupling_condition}
    &\left(
    \lambda_\phi 
    \int \d[d]{\bm k}
    \big|\tilde F_\textsc{a}(\bm k)\big|^2
    \right)^2
    \gg 
    \frac{1}{2}
    \int \d[d]{\bm k}\om k
    \big|\tilde F_\textsc{a}(\bm k)\big|^2,
    \\
    \label{eq:gamma_condition}
    &\gamma_\textsc{a}:=
    \lambda_\phi
    \lambda_\pi 
    \int\d[d]{\bm k} 
    \big|\tilde F_\textsc{a}(\bm k)\big|^2
    =
    \frac{\pi}{4}\mod 2\pi.
\end{align}
Here and throughout the text we use the notation where $\tilde g(\bm k)$ denotes the Fourier transform of the function $g(\bm x)$, defined by 
\begin{equation}
	\label{eq:FT}
    \tilde{g}(\bm{k})
	:=
	\frac{1}{\sqrt{(2\pi)^d}}\int\d[d]{\bm{x}}
	g(\bm{x})e^{\ii\bm{k}\cdot\bm{x}}.
\end{equation}

This claim is straightforwardly proven by direct calculation. Acting on the initial state $\left(c_1\ket{+_z}+c_2\ket{-_z}\right)\ket 0$ with the rightmost exponential in $\Ua$ results in the state
\begin{align}\label{entangled}
    c_1\ket{+_z}\ket{+\alpha_\textsc{a}}
    +c_2\ket{-_z}\ket{-\alpha_\textsc{a}},
\end{align}
where $\ket{\pm\alpha_\textsc{a}}$ are coherent field states defined by\footnote{For a comprehensive overview of coherent states of a scalar field and the notation used, see Refs.~\cite{Simidzija2017b,Simidzija2017c}}
\begin{equation}
\label{eq:Bell_state}
    \ket{\pm\alpha_\textsc{a}}
    :=
    \exp
    \big(
    \pm\ii\phia
    \big)
    \ket 0.
\end{equation}
It can be shown that the magnitude of the overlap between the two coherent states $\ket{+\alpha_\textsc{a}}$ and $\ket{-\alpha_\textsc{a}}$ is (see Appendix A of~\cite{Simidzija2017c})
\begin{equation}
    \left|
    \langle
    +\alpha_\textsc{a}|
    -\alpha_\textsc{a}
    \rangle
    \right|
    =
    \exp
    \left[
    -\left(\lambda_\phi\right)^2
    \int
    \frac{\d[d]{\bm k}}{2\om k}
    \big|\tilde F_\textsc{a}(\bm k)\big|^2
    \right].
\end{equation}
Hence we see that if $\lambda_\phi\gg1$ in units of the characteristic length scale set by $F_\textsc{a}(\bm x)$, then the field states $\ket{+\alpha_\textsc{a}}$ and $\ket{-\alpha_\textsc{a}}$ are almost orthogonal, and if $|c_1|=|c_2|=1/\sqrt{2}$ the state in Eq.~\eqref{eq:Bell_state} is almost maximally entangled. In other words, a stronger coupling between the qubit and the field results in a more correlated (i.e. entangled) state of the two systems.

The final step in evaluating the action of the unitary $\Ua$ on the initial state of the qubit-field system $\left(c_1\ket{+_z}+c_2\ket{-_z}\right)\ket 0$ is to apply the unitary $\exp(\ii\hat\sigma_x\pia)$ to the entangled state $c_1\ket{+_z}\ket{+\alpha_\textsc{a}}+c_2\ket{-_z}\ket{-\alpha_\textsc{a}}$. To perform this calculation, let us first apply the field observable $\pia$ to the coherent state $\ket{\pm\alpha_\textsc{a}}$. To that end, using the Baker-Campbell-Hausdorff lemma we can straightforwardly prove the identity
\begin{equation}
    \exp\big(\pm\ii\pia \big)\a k\exp\big(\mp\ii\pia\big)
    =
    \a k+\alpha_\textsc{a}(\bm k)\mathds 1,
\end{equation}
where the \textit{coherent amplitude} $\alpha_\nu(\bm k)$ is defined as
\begin{equation}
    \alpha_\nu(\bm k)
    =
    \frac{\lambda_\nu^\phi}{\sqrt{2\om k}}
    \tilde F^*_\nu(\bm k)e^{\ii\om k t_\nu},
\end{equation}
and hence we find that
\begin{equation}
\label{eq:eigenvalue_equation}
    \pia\ket{\pm\alpha_\textsc{a}}
    =
    \pm\gamma_\textsc{a}\ket{\pm\alpha_\textsc{a}}
    +
    \exp\big(\pm\ii\pia\big)\pia\ket{0},
\end{equation}
where $\gamma_\textsc{a}$ is defined in Eq.~\eqref{eq:gamma_condition}. Hence we see that $\pia\ket{\pm\alpha_\textsc{a}}$ is the sum of two terms, and in particular we find that if $\gamma_\textsc{a}^2\gg\bra{0}\pia^2\ket{0}$ then \begin{equation}
\label{eq:eigenvalue_equation_2}
    \pia\ket{\pm\alpha_\textsc{a}}
    \approx
    \pm\gamma_\textsc{a}\ket{\pm\alpha_\textsc{a}}.
\end{equation}
In other words, if $\gamma_\textsc{a}^2\gg\bra{0}\pia^2\ket{0}$ --- which is exactly equivalent to the condition~\eqref{eq:strong_coupling_condition} which we are assuming to hold --- then the field coherent states $\ket{\pm\alpha_\textsc{a}}$ are very approximately eigenstates of the observable $\pia$ with eigenvalues $\pm\gamma_\textsc{a}$. If additionally condition~\eqref{eq:gamma_condition} is satisfied, then we find
\begin{align}
    &\Ua \left(c_1\ket{+_z}+c_2\ket{-_z}\right)\ket 0\notag\\
    =\,&
    \exp\big(\ii\hat\sigma_x\pia\big)
    \left(
    c_1\ket{+_z}\ket{+\alpha_\textsc{a}}
    +c_2\ket{-_z}\ket{-\alpha_\textsc{a}}
    \right)
    \notag\\
    \approx\,&
    c_1 \exp\left(+\ii\frac{\pi}{4}\hat\sigma_x\right) \ket{+_z}\ket{+\alpha_\textsc{a}}
    \notag\\
    &\hspace{0.2cm}+
    c_2 \exp\left(-\ii\frac{\pi}{4}\hat\sigma_x\right) \ket{-_z}\ket{-\alpha_\textsc{a}}
    \notag\\
    =\,&
    \ket{+_y}
    \left(
    c_1\ket{+\alpha_\textsc{a}}
    -\ii c_2\ket{-\alpha_\textsc{a}}
    \right).
\end{align}
Note that in the second line we have used the identities $\exp(+\ii\frac{\pi}{4}\hat\sigma_x) \ket{+_z}=\ket{+_y}$ and $\exp(-\ii\frac{\pi}{4}\hat\sigma_x) \ket{+_z}=-\ii\ket{-_y}$, which simply state that we can perform Bloch sphere rotations of the eigenstates of $\hat\sigma_z$ into the positive eigenvalue eigenstate $\ket{+_y}$ of $\hat\sigma_y$ by applying rotation unitaries generated by $\hat\sigma_x$. Hence the unitary $\hat U_\textsc{a}$ has succeeded in encoding the orthogonal qubit superposition $c_1\ket{+_z}+c_2\ket{-_z}$ into an (almost) orthogonal superposition of coherent field states, $c_1\ket{+\alpha_\textsc{a}}-\ii c_2\ket{-\alpha_\textsc{a}}$.

The results of this section can be summarized as follows. An observer Alice coupling locally to a quantum field at a time $t_\textsc{a}$ can effectively encode her qubit into the field by implementing the unitary $\Ua=\exp\big(\ii\hat\sigma_x\pia\big)\exp\big(\ii\hat\sigma_z\phia\big)$, as long as the conditions \eqref{eq:strong_coupling_condition} and \eqref{eq:gamma_condition} are satisfied. For instance, if Alice's qubit starts in the equally weighted superposition $\frac{1}{\sqrt{2}}(\ket{+_z}+\ket{-_z})$, then, as long as \eqref{eq:strong_coupling_condition} is satisfied, the rightmost exponential in $\Ua$ will maximally entangle Alice's qubit with the field. Following this, and assuming that \eqref{eq:gamma_condition} is satisfied, the leftmost exponential in $\Ua$ will then use the state of the field to perform a controlled rotation in the Bloch sphere of the qubit, thus leaving the field in an equally weighted, orthogonal superposition of coherent states. In other words, the unitary $\hat U_\textsc{a}$ succeeds, through local operations, in encoding Alice's qubit into the field.

\subsection{Decoding a qubit out of a field}

Having understood how Alice can ENCODE her qubit of information into the field, the final step in constructing the field-mediated quantum channel from Alice to Bob, as depicted in Fig.~\ref{fig:encode_channel}, is to construct the DECODE gate that allows Bob to recover Alice's message from the field. The most straightforward way to proceed is to note that the DECODE gate should simply be the inverse of the ENCODE gate. Thus, since we know the unitary $\Ua=\exp\big(\ii\hat\sigma_x\pia\big)\exp\big(\ii\hat\sigma_z\phia\big)$ implementing the encode gate, we also know that the inverse unitary $\Ua^{-1}=\Ua^\dagger=\exp\big(-\ii\hat\sigma_z\phia\big)\exp\big(-\ii\hat\sigma_x\pia\big)$ will implement the DECODE gate. We can now simply set the unitary $\Ub$ in Fig.~\ref{fig:Xi}, which acts on detector B and the field, to be the unitary $\Ua^\dagger$ with the understanding that the qubit observables $\hat\sigma_x$ and $\hat\sigma_z$ now act on the Hilbert space $\mathcal H_\textsc{b}$ rather than $\mathcal H_\textsc{a}$.

Note however that there is a problem with this construction of the decoding unitary $\Ub$. Namely, while we have modified the qubit observables in $\Ub$ from the ones in $\Ua^\dagger$ so that now they act on $\mathcal H_\textsc{b}$ rather than $\mathcal H_\textsc{a}$, the field observables $\phia$ and $\pia$ appearing in $\Ub$ are still defined at the time $t_\textsc{a}$ (c.f. Eqs.~\eqref{eq:phia} and \eqref{eq:pia}). But in order for Bob to implement $\Ub$ at a later time $t_\textsc{b}$, he needs to couple his qubit to field observables defined at the time $t_\textsc{b}$, not at $t_\textsc{a}$. 

We will now solve this problem by proving a mathematical result which expresses the field observables $\phia$ and $\pia$ as observables at time $t_\textsc{b}$. Fundamentally, this result arises due to the fact that the field $\phih(\bm x,t)$ is by definition a solution to the wave equation, which, being a hyperbolic PDE, has a well defined initial value formulation that allows solutions at time $t_\textsc{a}$ to be propagated to solutions at time $t_\textsc{b}$. More concretely:

\begin{theorem}\label{theorem:AQFT}
Let $\phih(\bm x,t)$ be a free field in any spacetime dimension with mode expansion given by Eq.~\eqref{eq:field}. Let $\pih(\bm x,t)$ be the conjugate momentum field, and let $F(\bm x)$ be any smearing function. Then
\begin{align}
\label{eq:phi_AQFT_identity}
    \phih[F](\taa)
    &=
    \phih[F_2](\tbb)+\pih[F_1](\tbb), 
    \\
\label{eq:pi_AQFT_identity}
    \pih[F](\taa)
    &=
    \phih[F_3](\tbb)+\pih[F_2](\tbb),
\end{align}
where the $F_i(\bm x)$ are related to $F(\bm x)$ via their Fourier transforms as
\begin{align}
\label{eq:F1}
    \tilde F_1(\bm k)
    &=
    \tilde F(\bm k)\sinc(\Delta\om k)(-\Delta),\\
    \label{eq:F2}
    \tilde F_2(\bm k)
    &=
    \tilde F(\bm k)\cos(\Delta\om k),\\
    \label{eq:F3}
    \tilde F_3(\bm k)
    &=
    \tilde F(\bm k)\sin(\Delta\om k)\om k,
\end{align}
and where $\Delta:=\tbb-\taa$.
\end{theorem}

\begin{proof}
We will explicitly prove Eq.~\eqref{eq:phi_AQFT_identity}, while Eq.~\eqref{eq:pi_AQFT_identity} is proven analogously. Starting from the mode expansion for $\phih(\bm x,t)$ given in Eq.~\eqref{eq:field} we get:
\begin{align}
    \phih[F](\taa)
    &=\!\!\!
    \int\!\!\d[d]{\bm x}\!F(\bm x)\!\!
    \int\!\!\!\frac{\d[d]{\bm k}}{\sqrt{2(2\pi)^d\om k}}\!
    \left(\!\a k e^{-\ii(\om k \taa-\bm k\cdot\bm x)}\!+\!\text{H.c.}\!\right)\notag\\
    &=
    \int\frac{\d[d]{\bm k}}{\sqrt{2\om k}}
    \left(\tilde F(\bm k)\a k e^{-\ii\om k \taa}+\text{H.c.}\right),
\end{align}
with $\tilde F(\bm k)$ the Fourier transform of $F(\bm x)$ as defined in Eq.~\eqref{eq:FT}. Then, using $\Delta:=\tbb-\taa$ we obtain
\begin{align}
    \phih[F](\taa)
    &=
    \int\frac{\d[d]{\bm k}}{\sqrt{2\om k}}
    \left(\tilde F(\bm k)\a k e^{\ii\om k \Delta}e^{-\ii\om k\tbb}+\text{H.c.}\right)
    \notag\\
    &=
    \int\frac{\d[d]{\bm k}}{\sqrt{2\om k}}
    \Big(\tilde F(\bm k)
    \left[\cos(\om k\Delta)+\ii\sin(\om k\Delta)\right]
    \notag\\
    &\hspace{2cm}\times\a k e^{-\ii\om k\tbb}+\text{H.c.}\Big).
\end{align}
By introducing the Fourier transforms $\tilde F_1(\bm k)$ and $\tilde F_2(\bm k)$ as defined in Eqs.~\eqref{eq:F1} and \eqref{eq:F2}, we can write $\phih[F](\taa)$ as
\begin{align}
    &\int\frac{\d[d]{\bm k}}{\sqrt{2\om k}}
    \left(\tilde F_2(\bm k)\a k e^{-\ii\om k\tbb}+\text{H.c.}\right)\notag\\
    &\hspace{0.02cm}+
    \int\frac{\d[d]{\bm k}}{\sqrt{2\om k}}
    \left(-\ii\om k\tilde F_1(\bm k)\a k e^{-\ii\om k\tbb}+\text{H.c.}\right)
    \notag\\
    =\,&
    \!\!\!\int\!\!\d[d]{\bm x}F_2(\bm x)\!\!
    \int\!\!\!\frac{\d[d]{\bm k}}{\sqrt{2(2\pi)^d\om k}}\!
    \left(\!\a k e^{-\ii(\om k \tbb-\bm k\cdot\bm x)}\!+\!\text{H.c.}\!\right)
    \notag\\
    &\hspace{0.02cm}
    +
    \!\!\!\int\!\!\d[d]{\bm x}F_1(\bm x)\!\!
    \int\!\!\!\frac{\d[d]{\bm k}}{\sqrt{2(2\pi)^d\om k}}\!
    \left(\!-\ii\om k\a k e^{-\ii(\om k \tbb-\bm k\cdot\bm x)}\!+\!\text{H.c.}\!\right)
    \notag\\
    =\,&
    \phih[F_2](\tbb)+\pih[F_1](\tbb),
\end{align}
which proves Eq.~\eqref{eq:phi_AQFT_identity}.
\end{proof}

With this mathematical result at hand, we can now write the unitary $\Ub$ in Fig.~\ref{fig:Xi} --- which decodes Alice's qubit out of the field and onto Bob's detector --- in terms of field observables at the time $\tbb$. Namely, the theorem allows us to write the field observables $\phia$ and $\pia$ defined in Eqs.~\eqref{eq:phia} and \eqref{eq:pia} as
\begin{align}
    \phia
    &=
    \lambda_\phi
    \phih[F_{\textsc{b}2}](\tbb)
    +
    \lambda_\phi
    \pih[F_{\textsc{b}1}](\tbb),
    \notag\\
    \pia
    &=
    \lambda_\pi
    \phih[F_{\textsc{b}3}](\tbb)
    +
    \lambda_\pi
    \pih[F_{\textsc{b}2}](\tbb),
\end{align}
where Bob's smearing functions $F_{\textsc{b}i}(\bm x)$ are defined in terms of Alice's smearing $F_\textsc{a}$ through their Fourier transforms,
\begin{align}
\label{eq:FB1}
    \tilde F_{\textsc{b}1}(\bm k)
    &=
    \tilde F_\textsc{a}(\bm k)\sinc(\Delta\om k)(-\Delta),\\
    \label{eq:FB2}
    \tilde F_{\textsc{b}2}(\bm k)
    &=
    \tilde F_\textsc{a}(\bm k)\cos(\Delta\om k),\\
    \label{eq:FB3}
    \tilde F_{\textsc{b}3}(\bm k)
    &=
    \tilde F_\textsc{a}(\bm k)\sin(\Delta\om k)\om k.
\end{align}
Hence the unitary $\Ub$, defined by
\begin{align}
    \Ub
    &=
    \exp\big(-\ii\hat\sigma_z\phia\big)\exp\big(-\ii\hat\sigma_x\pia\big),
\end{align}
can now be alternatively defined in terms of field observables at time $\tbb$, namely
\begin{align}
    \Ub
    &=
    \exp
    \left[-\ii\lambda_\phi\hat\sigma_z\left(
    \phih[F_{\textsc{b}2}](\tbb)
    +
    \pih[F_{\textsc{b}1}](\tbb)\right)
    \right]\notag\\
    &\hspace{0.3cm}\times
    \exp
    \left[-\ii\lambda_\pi\hat\sigma_x\left(
    \phih[F_{\textsc{b}3}](\tbb)
    +
    \pih[F_{\textsc{b}2}](\tbb)\right)
    \right].
    \label{eq:Ub}
\end{align}

In summary, we have succeeded in constructing the quantum channel shown in Fig.~\ref{fig:Xi}, which allows Alice to perfectly transmit a qubit through a quantum field to Bob. The quantum channel consists of two steps: 
\begin{enumerate}[leftmargin=*]
    \item First, at time $t=\taa$, Alice encodes her qubit state in a spatial region of the field characterized by $F_\textsc{a}(\bm x)$ by implementing the unitary $\Ua$ given in Eq.~\eqref{eq:Ua}.
    \item Then, at a later time $t=\tbb$, Bob decodes the qubit from the field by coupling with the unitary $\Ub$ given in Eq.~\eqref{eq:Ub}. In order for Bob to be able to implement this unitary, his detector must be smeared in a spatial region that contains the supports of the functions $F_{\textsc{b}1}(\bm x)$, $F_{\textsc{b}2}(\bm x)$, and $F_{\textsc{b}3}(\bm x)$ defined by Eqs.~\eqref{eq:FB1}-\eqref{eq:FB3}.
\end{enumerate}
Additionally, in order for the channel to succeed, the conditions \eqref{eq:strong_coupling_condition} and \eqref{eq:gamma_condition} on the coupling strengths $\lambda_\phi$ and $\lambda_\pi$ must be satisfied. Physically, Eq.~\eqref{eq:strong_coupling_condition} is a strong-coupling condition which ensures that Alice's qubit first gets maximally entangled with orthogonal coherent field states, while Eq.~\eqref{eq:gamma_condition} is a fine-tuning condition which ensures that Alice's qubit is then rotated by the right amount in the Bloch sphere so that it gets completely unentangled from the field. Together, these conditions ensure that the encoding gate (and hence the decoding gate, which is just the inverse encoding gate) are implemented successfully. In particular we note that, as was discussed above, a strong (i.e. non-perturbative) coupling of detectors to the field is necessary in order for the field-mediated quantum channel from Alice to Bob to have maximal quantum channel capacity.

Despite our successes so far, there still remain two pertinent issues that must be addressed before one can be fully satisfied with our construction of a perfect, field-mediated quantum channel from Alice to Bob. First, it should be verified, without the use of any approximations (such as the one in Eq.~\eqref{eq:eigenvalue_equation_2}), that our supposedly perfect quantum channel $\Xi$ indeed has a maximal quantum channel capacity of $\mathcal Q(\Xi)=1$. And second, the smearing functions $F_{\textsc{b}i}(\bm x)$ are defined in terms of their Fourier transforms, and hence it is presently not clear where in space Bob needs to be located in order to receive Alice's quantum message, which is crucial for our study. We will successively address these two remaining issues in Sec.~\ref{sec:testing_quantum_channel} and Sec.~\ref{sec:where_does_q_info_propagate}.

\section{Numerical test of the perfect quantum channel}
\label{sec:testing_quantum_channel}

Let us verify that the channel $\Xi$ which we constructed in the previous section --- shown in Fig.~\ref{fig:Xi} with $\Ua$ and $\Ub$ given by Eqs.~\eqref{eq:Ua} and \eqref{eq:Ub} --- can indeed perfectly transmit quantum information from Alice to Bob. For convenience we will assume that Bob's initial state is $\ket{+_y}$, and that Alice's initial state (i.e. the input to the channel), is the maximally mixed state, $\hat\rho_{\textsc{a},0}=\frac{1}{2}\mathds 1$. As discussed in Sec.~\ref{sec:setup}, we will compute a lower bound on the quantum channel capacity $Q(\Xi)$ by computing the coherent information $I_c(\hat\rho_{\textsc{a},0},\Xi)$ of the channel $\Xi$ and the input state $\hat\rho_{\textsc{a},0}$.

Recall that to compute $I_c(\hat\rho_{\textsc{a},0},\Xi)$, we must first purify the input to the channel, i.e. the maximally mixed state $\hat\rho_{\textsc{a},0}$. To that end, we suppose that the initial state of Alice is entangled with some third qubit C, and that the joint state of C and Alice is given by the maximally entangled pure state $\ket{\psi_\textsc{ca}}=\frac{1}{\sqrt{2}}(\ket{-_z}\ket{+_z}+\ket{+_z}\ket{-_z})_\textsc{ca}$. 

Next, in order to compute $I_c(\hat\rho_{\textsc{a},0},\Xi)$, we must evaluate the state 
\begin{equation}
    \hat\rho_\textsc{cb}
    :=
    \left(
    \mathds 1_\textsc{c}\otimes\Xi
    \right)
    \left(
    \ket{\psi_\textsc{ca}}\bra{\psi_\textsc{ca}}
    \right),
\end{equation}
on $\mathcal H_\textsc{c}\otimes \mathcal H_\textsc{b}$. Following this we can easily compute the coherent information through Eq.~\eqref{eq:coherent_info}, i.e. as $I_c(\hat\rho_{\textsc{a},0},\Xi):=S(\hat\rho_\textsc{b})-S(\hat\rho_\textsc{cb})$. Writing $\Xi$ in terms of the unitaries $\Ua$ and $\Ub$, we obtain
\begin{equation}
\label{eq:rho_cb}
    \hat\rho_\textsc{cb}
    =
    \tr_{\textsc{a}\phi}\!
    \left[
    \Ub^\dagger\Ua^\dagger
    \left(
    \ket{\psi_\textsc{ca}}\bra{\psi_\textsc{ca}}
    \!
    \otimes
    \!
    \ket{0}\bra{0}
    \!
    \otimes
    \!
    \ket{+_y}\bra{+_y}
    \right)
    \hat U_\textsc{a}^\dagger
    \hat U_\textsc{b}^\dagger
    \right]\!.
\end{equation}
The simplest way to proceed with the computation of this density matrix is to decompose the unitaries $\Ua$ and $\Ub$ into products of controlled unitaries from the qubits A and B onto the field. Namely, for $\Ua$ we write
\begin{align}
\label{eq:Ua_controlled}
    \Ua
    &=
    \exp\big(\ii\hat\sigma_x\pia\big)
    \exp\big(\ii\hat\sigma_z\phia\big)
    \notag\\
    &=
    \sum_{x,z\in\{\pm\}}
    \P x\P z \otimes e^{\ii x\pia}e^{\ii z\phia},
\end{align}
where $\P x$ and $\P z$ are the projectors onto the eigenstates of $\hat\sigma_x$ and $\hat\sigma_z$ (note that to simplify notation we are using the dummy summation index $x$ or $z$ on the $\hat P$ to denote what operator the projector is associated with). Written in this form we see that the action of $\Ua$ is to unitarily evolve the field state with a unitary that is dependent on the outcome of a $\hat\sigma_z$ measurement of the qubit A, and then to do the same thing for a $\hat\sigma_x$ measurement. In other words $\Ua$ is a product of two controlled unitaries, from A to the field.

We can perform the same kind of decomposition for the unitary $\Ub$ by starting with the expression Eq.~\eqref{eq:Ub}. However, it is more convenient to write $\Ub$ in the way it was initially defined, i.e. as $\Ua^{-1}=\Ua^\dagger$ with the understanding that the qubit observables are now observables on $\mathcal H_\textsc{b}$ rather than $\mathcal H_\textsc{a}$. Hence, from Eq.~\eqref{eq:Ua_controlled} we directly obtain
\begin{align}
\label{eq:Ub_controlled}
    \Ub
    &=
    \sum_{x,z\in\{\pm\}}
    \P z \P x \otimes e^{-\ii z\phia}e^{-\ii x\pia},
\end{align}
where the self-adjoint projectors $\P x$ and $\P z$ are associated with the Pauli operators $\hat\sigma_x$ and $\hat\sigma_z$ on $\mathcal H_\textsc{b}$.

Substituting Eqs.~\eqref{eq:Ua_controlled} and \eqref{eq:Ub_controlled} for $\Ua$ and $\Ub$ into Eq.~\eqref{eq:rho_cb} for $\hat\rho_\textsc{cb}$, and writing $\ket{\psi_\textsc{ca}}=\frac{1}{\sqrt{2}}\sum_j\ket{-j_z}\ket{j_z}$ with $j\in\{\pm\}$, we get
\begin{align}
    \hat\rho_\textsc{cb}
    &=
    \frac{1}{2}
    \sum_{j,k,x_i,z_i}
    \bra{0}
    e^{-\ii z_1\phia}
    e^{-\ii x_1\pia}
    e^{\ii x_2\pia}
    e^{\ii z_2\phia}\notag\\
    &\hspace{1cm}\times 
    e^{-\ii z_3\phia}
    e^{-\ii x_3\pia}
    e^{\ii x_4\pia}
    e^{\ii z_4\phia}
    \ket{0}
    \notag\\
    &\hspace{1cm}\times
    \bra{k_z}\P{z_1}\P{x_1}\P{x_4}\P{z_4}\ket{j_z}_\textsc{a}
    \ket{-j_z}_\textsc{c}\bra{-k_z}\notag\\
    &\hspace{1cm}\otimes
    \P{z_3}\P{x_3}\ket{+_y}_\textsc{b}\bra{+_y}
    \P{x_2}\P{z_2},
\end{align}
where $x_i$ stands for $x_1,x_2,x_3,x_4$, and similarly for $z_i$, and where all of the summation variables run over the set $\{+1,-1\}$, such that there are $2^{10}$ terms in the entire sum. This expression can straightforwardly be evaluated by a computer as long as we can first simplify the field expectation value $\bra{0}\dots\ket{0}$. In order to do so, let us first redefine the summation indices by $x_1\mapsto -x_1, z_1\mapsto -z_1, x_3\mapsto -x_3$ and $z_3\mapsto -z_3$, such that the expression for $\hat\rho_\textsc{cb}$ reads
\begin{align}
    \hat\rho_\textsc{cb}
    &=
    \frac{1}{2}
    \sum_{j,k,x_i,z_i}
    \bra{0}
    e^{\ii z_1\phia}
    e^{\ii x_1\pia}
    e^{\ii x_2\pia}
    e^{\ii z_2\phia}\notag\\
    &\hspace{1cm}\times
    e^{\ii z_3\phia}
    e^{\ii x_3\pia}
    e^{\ii x_4\pia}
    e^{\ii z_4\phia}
    \ket{0}
    \notag\\
    &\hspace{1cm}\times
    \bra{k_z}\P{-z_1}\P{-x_1}\P{x_4}\P{z_4}\ket{j_z}_\textsc{a}
    \ket{-j_z}_\textsc{c}\bra{-k_z}\notag\\
    &\hspace{1cm}\otimes
    \P{-z_3}\P{-x_3}\ket{+_y}_\textsc{b}\bra{+_y}
    \P{x_2}\P{z_2}.
\end{align}

Using the Baker-Campbell-Hausdorff formula we can write~\cite{Truax1988}
\begin{align}
    e^{\ii z_i\phia}e^{\ii x_i\pia}
    =
    e^{x_i z_i C}e^{\ii\OO i},
\end{align}
where $\OO i:= x_i\pia+z_i\phia$, and $C$ is defined by
\begin{align}
    C:=-\frac{1}{2}\langle[\phia,\pia]\rangle =
    \label{eq:C}
    -\frac{\ii\lambda_\phi\lambda_\pi}{2}
    \int\d[d]{\bm k}|\tilde F_\textsc{a}(\bm k)|^2,
\end{align}
where we note that the commutator in \eqref{eq:C} is proportional to the identity and thus its expectation value is state-independent. Then, $\hat\rho_\textsc{cb}$ can be written as
\begin{align}
    \hat\rho_\textsc{cb}
    &=
    \frac{1}{2}
    \sum_{j,k,x_i,z_i}
    e^{x_1 z_1 C}
    e^{-x_2 z_2 C}
    e^{x_3 z_3 C}
    e^{-x_4 z_4 C}
    \notag\\
    &\hspace{1cm}\times
    \bra{0}
    e^{\ii\OO 1}
    e^{\ii\OO 2}
    e^{\ii\OO 3}
    e^{\ii\OO 4}
    \ket{0}
    \notag\\
    &\hspace{1cm}\times
    \bra{k_z}\P{-z_1}\P{-x_1}\P{x_4}\P{z_4}\ket{j_z}_\textsc{a}
    \ket{-j_z}_\textsc{c}\bra{-k_z}\notag\\
    &\hspace{1cm}\otimes
    \P{-z_3}\P{-x_3}\ket{+_y}_\textsc{b}\bra{+_y}
    \P{x_2}\P{z_2}.
\end{align}

To further simplify this expression for $\hat\rho_\textsc{cb}$ let us make use of the identity
\begin{equation}
    \bra{0}
    \prod_{l=1}^n
    e^{\ii\OO j}
    \ket{0}
    =
    \prod_{l<m}^n e^{-W_{lm}}
    \prod_{l=1}^n e^{-\frac{1}{2}W_{ll}},
\end{equation}
where $W_{lm}:=\bra{0}\OO l\OO m\ket{0}$. This identity holds for any operators $\OO j$ which are linear in the field creation and annihilation operators, and it can straightforwardly be proven using Wick's theorem~\cite{Peskin1996}.\footnote{In fact, an analogous version of this identity holds not only for the field vacuum state $\ket 0$, but also for any Gaussian state~\cite{Adesso2014}.} Then $\hat\rho_\textsc{cb}$ becomes
\begin{align}
\label{eq:rho_cb_2}
    \hat\rho_\textsc{cb}
    &=
    \frac{1}{2}
    \sum_{j,k,x_i,z_i}
    e^{x_1 z_1 C_1}
    e^{-x_2 z_2 C_2}
    e^{x_3 z_3 C_3}
    e^{-x_4 z_4 C_4}\notag\\
    &\hspace{1cm}\times
    \prod_{l<m}^4 e^{-W_{lm}}
    \prod_{l=1}^4 e^{-\frac{1}{2}W_{ll}}
    \notag\\
    &\hspace{1cm}\times
    \bra{k_z}\P{-z_1}\P{-x_1}\P{x_4}\P{z_4}\ket{j_z}_\textsc{a}
    \ket{-j_z}_\textsc{c}\bra{-k_z}\notag\\
    &\hspace{1cm}\otimes
    \P{-z_3}\P{-x_3}\ket{+_y}_\textsc{b}\bra{+_y}
    \P{x_2}\P{z_2},
\end{align}
and $W_{lm}$ evaluates to
\begin{align}
    W_{lm}
    &=
    \int\frac{\d[d]{\bm k}}{2|\bm k|}
    |\tilde F_\textsc{a}(\bm k)|^2
    \left(
    z_l\lambda_\phi-\ii|\bm k|x_l\lambda_\pi
    \right)
    \notag\\
    &\hspace{1cm}\times\left(
    z_m\lambda_\phi-\ii|\bm k|x_m\lambda_\pi
    \right).
    \label{eq:W_lm}
\end{align}
Hence we now see that if we specify the coupling constants $\lambda_\phi$ and $\lambda_\pi$, as well as the smearing function $\hat F_\textsc{a}(\bm x)$, we can straightforwardly compute $C$ and $W_{lm}$ (at least numerically), and hence obtain a (numerical) result for the density matrix $\hat\rho_\textsc{cb}$.

\subsection{Gaussian detector smearing}
\label{sec:gaussian}

In order to numerically compute the coherent information of our quantum channel, let us now particularize our discussion to $(3+1)$-dimensions, and let us set the smearing function of Alice's detector to be a Gaussian of width $\sigma$, i.e.
\begin{equation}
    F_\textsc{a}(\bm x)=\frac{1}{(\sqrt{\pi}\sigma)^3}\exp\left(-\frac{|\bm x|^2}{\sigma^2}\right),
\end{equation}
which has a Fourier transform that is given by
\begin{equation}
    \tilde F_\textsc{a}(\bm k)=\frac{1}{\sqrt{(2\pi)^3}}\exp\left(-\frac{1}{4}|\bm k|^2\sigma^2\right).
\end{equation}
Then, the conditions \eqref{eq:strong_coupling_condition} and \eqref{eq:gamma_condition} on the coupling strengths $\lambda_\phi$ and $\lambda_\pi$, which we require in order to have a perfect quantum channel, simplify to
\begin{align}
\label{eq:strong_coupling_condition_2}
    &\lambda_\phi\gg\sigma, \text{ and}
    \\
    \label{eq:gamma_condition_2}
    &\frac{\lambda_\phi\lambda_\pi}{\sqrt{(2\pi)^3}\sigma^3}
    =
    \frac{\pi}{4}\mod 2\pi.
\end{align}
In particular, recalling that the strong-coupling condition \eqref{eq:strong_coupling_condition} is a requirement in order for Alice's qubit to become maximally entangled with the field, we find that this is only possible if the coupling strength $\lambda_\phi$ of the detector is much larger than its size. Finally, using Eqs.~\eqref{eq:C} and \eqref{eq:W_lm}, we can also readily obtain simple analytical expressions for $C$ and $W_{lm}$, namely
\begin{align}
    C 
    &= 
    \frac{-\ii\lambda_\phi\lambda_\pi}{2\sqrt{(2\pi)^3}\sigma^3},\\
    W_{lm} 
    &=
    \frac{1}{8\pi^2\sigma^4}
    \big[
    4 x_l x_m \lambda_\pi^2+2 z_l z_m \sigma^2\lambda_\phi^2
    \notag\\
    &\hspace{1.5cm}
    +\ii\sqrt{2\pi}\sigma\lambda_\phi\lambda_\pi\left(x_m z_l-x_l z_m\right)\big].
\end{align}
We now have all of the necessary components to compute $\hat\rho_\textsc{cb}$ via Eq.~\eqref{eq:rho_cb_2}, and hence to compute the coherent information $I_c(\hat\rho_{\textsc{a},0},\Xi)$ of the channel $\Xi$ and the input state $\hat\rho_{\textsc{a},0}$. 

\begin{figure}
    \centering
    \includegraphics[width=\linewidth]{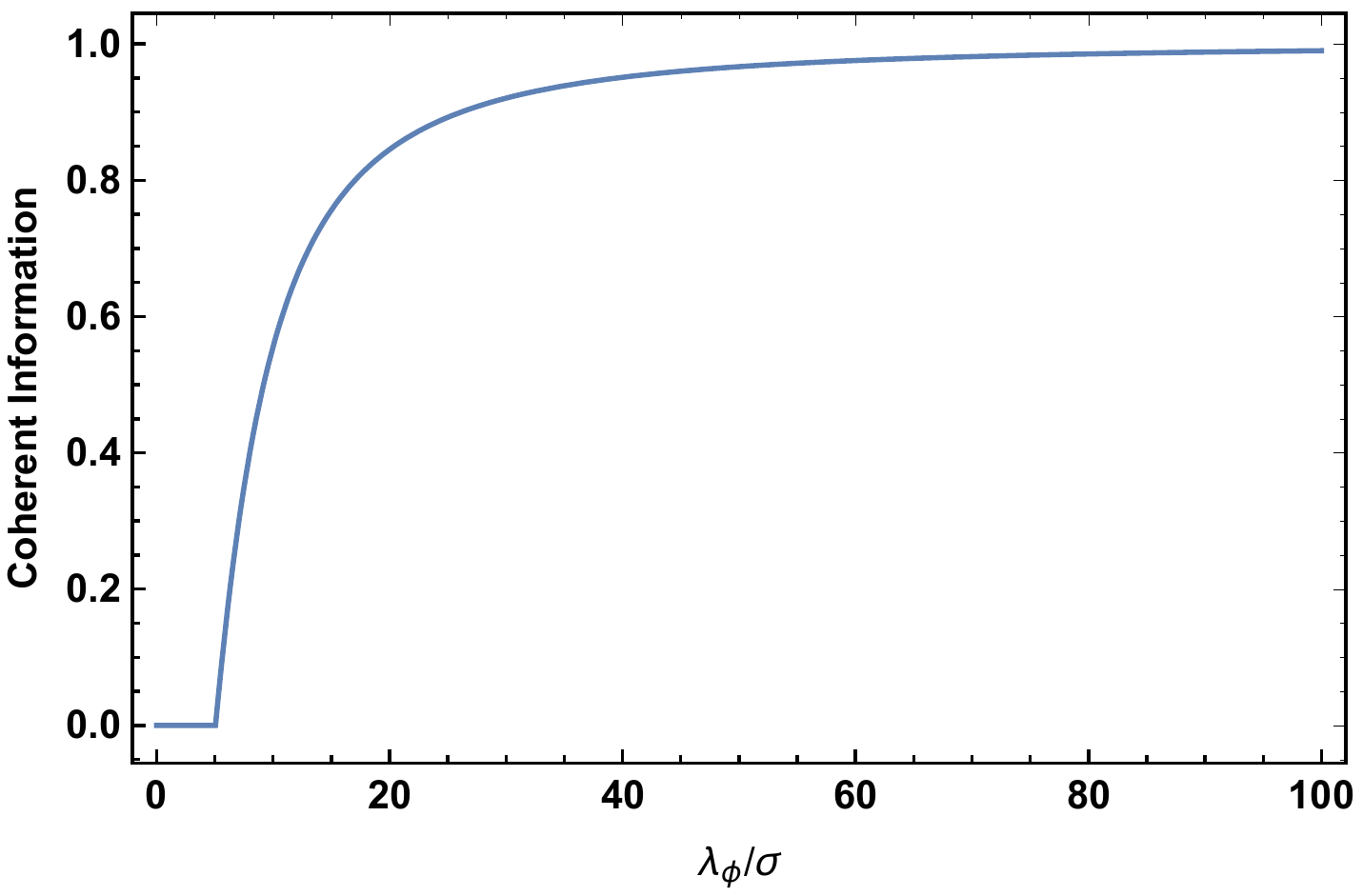}
    \caption{Plot of $\max\{0,I_c(\hat\rho_{\textsc{a},0},\Xi)\}$ versus the ratio of the coupling strength $\lambda_\phi$ to the size of Alice's detector, $\sigma$. Here $\hat\rho_{\textsc{a},0}=\frac{1}{2}\mathds 1$ is the maximally mixed input state to the channel $\Xi$. Notice that for $\lambda_\phi\gg\sigma$, the coherent information approaches its maximum possible value of 1, thus confirming that in this limit $\Xi$ is a perfect quantum channel.}
    \label{fig:C_vs_lambda}
\end{figure}

In Fig.~\ref{fig:C_vs_lambda} we plot $I_c(\hat\rho_{\textsc{a},0},\Xi)$ versus $\lambda_\phi/\sigma$, and we find, as expected that for $\lambda_\phi/\sigma\rightarrow\infty$, the coherent information $I_c(\hat\rho_{\textsc{a},0},\Xi)$ approaches its maximum value of 1. Additionally, since we know that the quantum channel capacity $Q(\Xi)$ is lower bounded by $I_c(\hat\rho_{\textsc{a},0},\Xi)$, and since we also know that $Q(\Xi)\le 1$ (i.e. a single use of the channel can transmit at most one qubit), we thus conclude that in the limit $\lambda_\phi/\sigma\rightarrow\infty$, the quantum channel capacity $Q(\Xi)$ approaches its maximum value of 1. In other words, we have verified, without the use of any approximations, that the field-mediated quantum channel from Alice to Bob is indeed a perfect quantum channel if the conditions \eqref{eq:strong_coupling_condition} and \eqref{eq:gamma_condition} are satisfied.

\section{Where does the quantum information propagate?}
\label{sec:where_does_q_info_propagate}

While we have mathematically verified that the quantum channel $\Xi$ from Alice to Bob is a perfect quantum channel, we still have some work to do in order to understand the physics of quantum information propagation through a relativistic quantum field. In particular, we have yet to discuss where in space Bob needs to be located at time $\tbb$ in order to receive Alice's message, which she encoded in the field at an earlier time $\taa$. Let us now attempt to better understand this issue.

Recall from Theorem~\ref{theorem:AQFT} that if Alice couples to the field at time $\taa$ with a spatial smearing $F_\textsc{a}(\bm x)$, then in order for Bob to perfectly recover Alice's message at a time $\tbb$ he needs to be able to couple his detector to the field $\phih$ and the conjugate field $\pih$ with three different smearing functions $F_{\textsc{b}i}(\bm x)$, which are related to $F_\textsc{a}(\bm x)$ via their Fourier transforms,
\begin{align}
\label{eq:FB1_b}
    \tilde F_{\textsc{b}1}(\bm k)
    &=
    \tilde F_\textsc{a}(\bm k)\sinc(\Delta\om k)(-\Delta),\\
    \label{eq:FB2_b}
    \tilde F_{\textsc{b}2}(\bm k)
    &=
    \tilde F_\textsc{a}(\bm k)\cos(\Delta\om k),\\
    \label{eq:FB3_b}
    \tilde F_{\textsc{b}3}(\bm k)
    &=
    \tilde F_\textsc{a}(\bm k)\sin(\Delta\om k)\om k,
\end{align}
where $\Delta:=\tbb-\taa$. Also recall that this result is valid in a flat spacetime of any dimension, and for any field mass. However, because the inverse Fourier transform is different in different spacetime dimensions, we expect the coordinate space functions $F_{\textsc{b}i}(\bm x)$ to have significantly different forms in different spacetimes. To see that this is indeed the case, let us now consider the $(3+1)$ and $(2+1)$ dimensional cases, both with a massless field. We will find that quantum information propagates very differently through the relativistic field in these two spacetimes, and that the violations of strong Huygens principle \cite{McLenaghan1974} play a fundamental role in the localization of quantum information encoded in quantum fields.

\subsection{(3+1)-dimensions}

In $(3+1)$-dimensions we are fortunate enough that we can obtain very simple and intuitive expressions for Bob's smearing functions $F_{\textsc{b}i}(\bm x)$, which are related to Alice's smearing function $F_\textsc{a}(\bm x)$ via their Fourier transforms via Eqs.~\eqref{eq:FB1_b}-\eqref{eq:FB3_b}. 

To obtain these expressions for $F_{\textsc{b}i}(\bm x)$, let us first recall the $d$-dimensional convolution theorem~\cite{Arfken1999}, which states that for two functions $f,g\in L^1(\mathbb R^d)$, 
\begin{align}
    \mathcal F^{-1}
    \left[\mathcal F[f]\mathcal F[g]
    \right](\bm x)
    &=
    \frac{1}{\sqrt{(2\pi)^d}}
    \left(
    f*g
    \right)(\bm x)\\
    &:=
    \frac{1}{\sqrt{(2\pi)^d}}
    \int\d[d]{\bm x'}f(\bm x')g(\bm x-\bm x'),\notag
\end{align}
where $\mathcal F$ denotes the $d$-dimensional Fourier transform defined by Eq.~\eqref{eq:FT}, $\mathcal F^{-1}$ denotes the inverse Fourier transform, and $(f*g)(\bm x)$ is the convolution product between $f(\bm x)$ and $g(\bm x)$.

Applying the convolution theorem to Eqs.~\eqref{eq:FB1_b}-\eqref{eq:FB3_b}, we obtain
\begin{align}
\label{eq:FB1_c}
    F_{\textsc{b}1}(\bm x)
    &=
    \frac{1}{\sqrt{(2\pi)^d}}
    \int\d[d]{\bm x'}
    F_\textsc{a}(\bm x')\notag\\
    &\hspace{1.5cm}\times
    \mathcal F^{-1}[-\Delta\sinc(\Delta |\bm k|)](\bm x-\bm x'),\\
    \label{eq:FB2_c}
    F_{\textsc{b}2}(\bm x)
    &=
    \frac{1}{\sqrt{(2\pi)^d}}
    \int\d[d]{\bm x'}
    F_\textsc{a}(\bm x')\notag\\
    &\hspace{1.5cm}\times
    \mathcal F^{-1}[\cos(\Delta |\bm k|)](\bm x-\bm x'),\\
    \label{eq:FB3_c}
    F_{\textsc{b}3}(\bm x)
    &=
    \frac{1}{\sqrt{(2\pi)^d}}
    \int\d[d]{\bm x'}
    F_\textsc{a}(\bm x')\notag\\
    &\hspace{1.5cm}\times
    \mathcal F^{-1}[|\bm k|\sin(\Delta |\bm k|)](\bm x-\bm x'),
\end{align}
where we note that $\om k=|\bm k|$ since we are setting the field mass $m$ equal to zero. Note that the above expressions are valid for any spatial dimension $d$ of the flat spacetime. In order to proceed to calculate $F_{\textsc{b}i}(\bm x)$ via Eqs.~\eqref{eq:FB1_c}-\eqref{eq:FB3_c}, we must compute the inverse Fourier transforms of the functions $-\Delta\sinc(\Delta |\bm k|),\cos(\Delta|\bm k|)$, and $|\bm k|\sin(\Delta |\bm k|)$. 

Let us now particularize to $d=3$, in which case obtaining explicit (distributional) expressions for these inverse Fourier transforms is possible. Namely we find
\begin{align}
\label{eq:FT_sinc}
    \mathcal F_{d=3}^{-1}[-\Delta\sinc(\Delta |\bm k|)](\bm x)
    &=
    -\sqrt{(2\pi)^3}\frac{\delta(r-\Delta)}{4\pi r},\\
    \label{eq:FT_cos}
    \mathcal F_{d=3}^{-1}[\cos(\Delta |\bm k|)](\bm x)
    &=
    -\sqrt{(2\pi)^3}\frac{\delta'(r-\Delta)}{4\pi r},\\
    \label{eq:FT_sin}
    \mathcal F_{d=3}^{-1}[|\bm k|\sin(\Delta |\bm k|)](\bm x)
    &=
    -\sqrt{(2\pi)^3}\frac{\delta''(r-\Delta)}{4\pi r},
\end{align}
with $r:=|\bm x|$ and where we are explicitly indicating that these are 3-dimensional inverse Fourier transforms. Here, $\delta'(\bm x)$ and $\delta''(\bm x)$ denote the first and second derivatives of the delta function. It is easiest to verify these results by taking the Fourier transform of the right-hand sides and checking that we get the expected answer. For example, let us verify Eq.~\eqref{eq:FT_cos} in this way. We find
\begin{align}
    &\mathcal F_{d=3}
    \left[
    -\sqrt{(2\pi)^3}\frac{\delta'(r-\Delta)}{4\pi r}
    \right](\bm k)\notag\\
    =\,&
    -\int\d[3]{\bm x}
    \frac{\delta'(r-\Delta)}{4\pi r}
    e^{-\ii\bm k\cdot\bm x}\notag\\
    =\,&
    -4\pi\int_0^\infty
    \dif r\, r^2 
    \frac{\delta'(r-\Delta)}{4\pi r}
    \frac{\sin(r|\bm k|)}{r|\bm k|},\notag\\
    =\,&
    -\frac{1}{|\bm k|}\int_0^\infty
    \dif r\,
    \delta'(r-\Delta)
    \sin(r|\bm k|)\\
    =\,&
    \cos(\Delta|\bm k|),
\end{align}
which proves Eq.~\eqref{eq:FT_cos}. Eqs.~\eqref{eq:FT_sinc} and \eqref{eq:FT_sin} can be proven analogously.

Substituting Eqs.~\eqref{eq:FT_sinc}-\eqref{eq:FT_sin} into Eqs.~\eqref{eq:FB1_c}-\eqref{eq:FB3_c}, we find that in $(3+1)$-dimensions Bob's smearing functions $F_{\textsc{b}i}(\bm x)$ are given in terms of Alice's smearing $F_\textsc{a}(\bm x)$ as
\begin{align}
\label{eq:FB1_d}
    F_{\textsc{b}1}(\bm x)
    &=
    -
    \int\d[d]{\bm x'}
    F_\textsc{a}(\bm x')
    \frac{\delta(|\bm x-\bm x'|-\Delta)}{4\pi |\bm x-\bm x'|},\\
    \label{eq:FB2_d}
    F_{\textsc{b}2}(\bm x)
    &=
    -
    \int\d[d]{\bm x'}
    F_\textsc{a}(\bm x')
    \frac{\delta'(|\bm x-\bm x'|-\Delta)}{4\pi |\bm x-\bm x'|},\\
    \label{eq:FB3_d}
    F_{\textsc{b}3}(\bm x)
    &=
    -
    \int\d[d]{\bm x'}
    F_\textsc{a}(\bm x')
    \frac{\delta''(|\bm x-\bm x'|-\Delta)}{4\pi |\bm x-\bm x'|}.
\end{align}
Hence, since the $\delta$, $\delta'$ and $\delta''$ functions above only have support if $|\bm x-\bm x'|=\Delta$, we find that in $(3+1)$-dimensions Bob's smearing functions $F_{\textsc{b}i}(\bm x)$ on the time-slice \mbox{$t=\tbb=\taa+\Delta$} only have support if they are in lightlike separation from Alice's smearing $F_\textsc{a}(\bm x)$ on the time-slice $t=\taa$. Therefore, in order for Bob to fully receive Alice's quantum message ($Q(\Xi)\to 1)$ through our field-mediated quantum channel in $(3+1)$-dimensions, he needs to be able to couple his detector on the entirety of Alice's lightcone. In other words, quantum information in $(3+1)$-dimensions propagates through a massless field precisely at the speed of light. While this result may be intuitive, we will see that it is actually a peculiarity of odd-spatial dimensions flat spacetime (e.g., 3+1 dimensional Minkowski space) and will not be true in general.

To conclude this section, let us illustrate the above result by considering a particular smearing $F_\textsc{a}(\bm x)$ for Alice, namely the Gaussian function considered in Sec.~\ref{sec:gaussian},
\begin{equation}
\label{eq:FA_gaussian}
    F_\textsc{a}(\bm x)=\frac{1}{(\sqrt{\pi}\sigma)^3}\exp\left(-\frac{|\bm x|^2}{\sigma^2}\right).
\end{equation}
We plot $F_\textsc{a}(\bm x)$ and the resulting smearing functions $F_{\textsc{b}i}(\bm x)$ for Bob's detector in Fig.~\ref{fig:Smearings_3D}. Notice that, as expected, at time $\tbb=\taa+\Delta$ Bob needs to couple to the field only near $|\bm x|=\Delta$ (i.e. on Alice's lightcone) in order to be able to fully recover her quantum message.

\begin{figure}
    \centering
    \includegraphics[width=\linewidth]{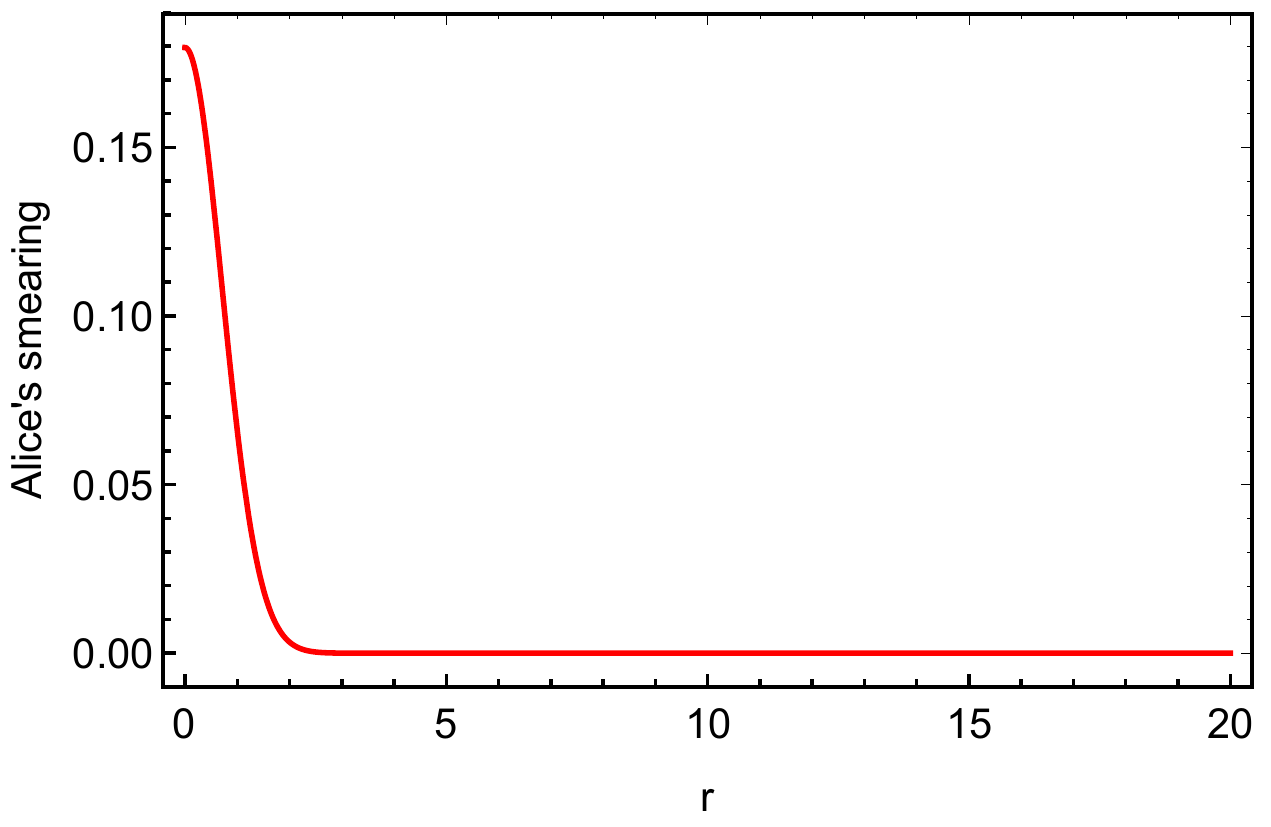}
    \includegraphics[width=\linewidth]{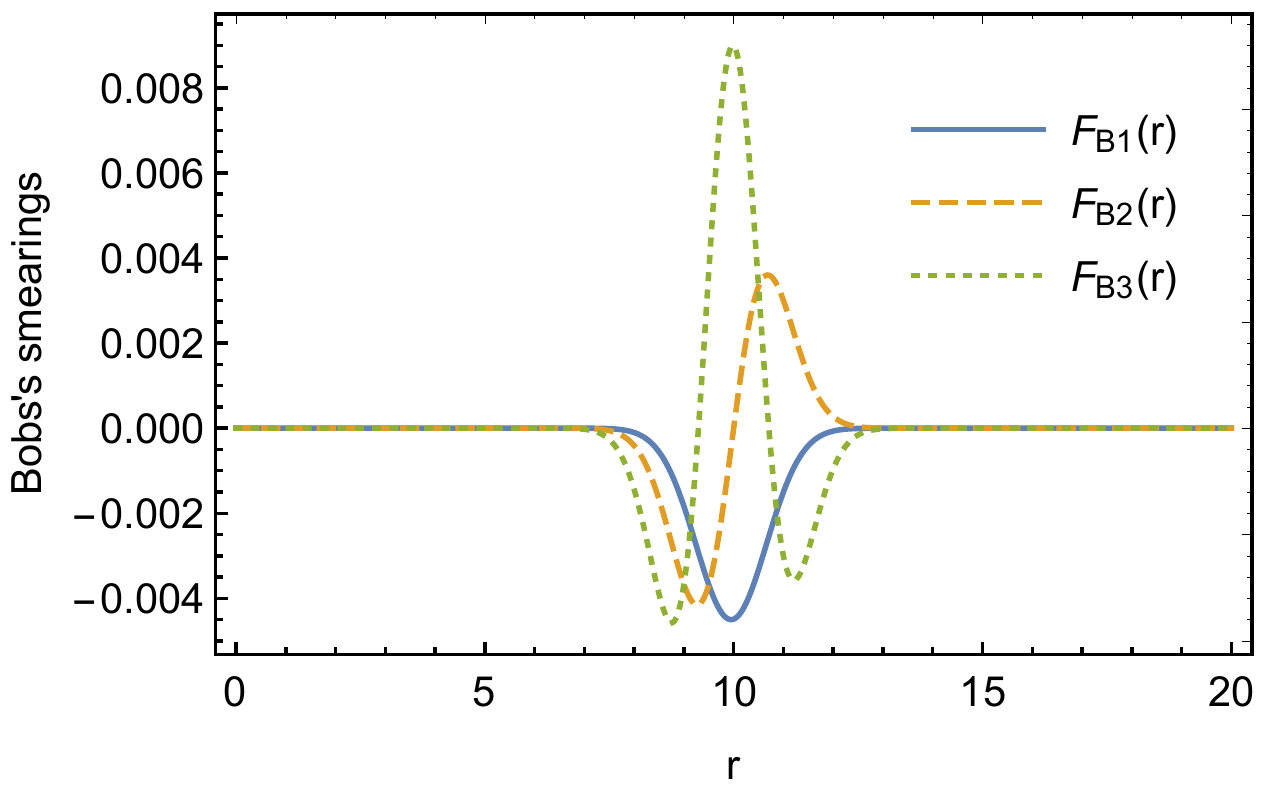}
    \caption[Propagation of quantum information through a massless field in $(3+1)$-dimensional flat spacetime]{$(3+1)$-dimensions. \textbf{Top}: Location at which Alice couples to the massless field at time $\taa$, given by Eq.~\eqref{eq:FA_gaussian}, with $\sigma=1$ and where $r=|\bm x|$ is measured in units of $\sigma$. \textbf{Bottom}: Bob's smearing functions, which dictate where in space Bob needs to couple to the field at time $\tbb=\taa+\Delta$ in order to receive Alice's message (we set $\Delta =10$).}
    \label{fig:Smearings_3D}
\end{figure}

The main result of this section --- i.e. that Bob needs to be lightlike separated from Alice in $(3+1)$-dimensions in order to receive her quantum message --- is fundamentally related to the strong Huygens principle, which we recall from our discussion in the introduction holds in $(3+1)$-dimensional flat spacetime. Namely, recall that the strong Huygens principle states that the massless field's radiation Green's function (and hence the expectation of the field commutator in the quantum case) only has support between lightlike separated events~\cite{McLenaghan1974}, and hence communication between observers via this quantum field is only possible if they are in null separation. While the implications of this fact have previously been studied in great detail for \textit{classical} communication protocols~\cite{Jonsson2015,Blasco2015,Jonsson2016,Blasco2016,Simidzija2017}, the work presented here is the first time that the effects of the strong Huygens principle have been studied in the context of \textit{quantum} communication.

\subsection{(2+1)-dimensions}

Let us now attempt to repeat the analysis of the previous section, but this time in $(2+1)$-dimensional Minkowski space. We expect to find significant differences to the $(3+1)$-dimensional case, due to the violations of the strong Huygens principle that occur in the former but not the latter spacetime.

Recall that the key expressions directly relating Bob's smearing functions $F_{\textsc{b}i}(\bm x)$ to Alice's smearing function $F_\textsc{a}(\bm x)$ are Eqs.~\eqref{eq:FB1_c}-\eqref{eq:FB3_c}. As in the $(3+1)$-dimensional case, in order to gain insight into the propagation of quantum information from these equations, we must first compute the Fourier transforms of the functions $-\Delta\sinc(\Delta |\bm k|)$, $\cos(\Delta|\bm k|)$, and $|\bm k|\sin(\Delta |\bm k|)$. Unfortunately however, we are only aware of a closed form expression for the first of these Fourier transforms, which reads
\begin{equation}
    \label{eq:FT_sinc_2D}
    \mathcal F_{d=2}^{-1}[-\Delta\sinc(\Delta |\bm k|)](\bm x)
    =
    \begin{cases}
    \frac{1}{\Delta\sqrt{\Delta^2-r^2}}
    & r<\Delta,\\
    0 & r\ge \Delta,
    \end{cases}
\end{equation}
and where once again $r:=|\bm x|$. Nevertheless, from this equation alone we can see an interesting feature of the propagation of quantum information in $(2+1)$-dimensions. Namely, unlike the 3D Fourier transforms given by Eqs.~\eqref{eq:FT_sinc}-\eqref{eq:FT_sin}, which only had support for $r=\Delta$, the 2D Fourier transform of $\sinc(\Delta|\bm k|)$ has support inside the light cone, i.e. for $r<\Delta$. Hence, after inserting this Fourier transfrom into Eq.~\eqref{eq:FB1_c}, we find that in $(2+1)$-dimensions the first of Bob's smearing functions, $F_{\textsc{b}1}(\bm x)$, is given by
\begin{equation}
\label{eq:FB1_2D}
    F_{\textsc{b}1}(\bm x)
    =
    \frac{1}{2\pi}
    \int_{B_\Delta(\bm x)}\d[2]{\bm x'}
    \frac{F_\textsc{a}(\bm x')}{\Delta\sqrt{\Delta^2-|\bm x-\bm x'|^2}},
\end{equation}
where $B_\Delta(\bm x)$ is the ball of radius $\Delta$ centered at $\bm x$. Thus we see that $F_{\textsc{b}1}(\bm x)$ has support even if $|\bm x-\bm x'|<\Delta$, and hence we conclude that if Bob wants to receive all possible quantum information from Alice in $(2+1)$-dimensions, then he needs to have access not only to Alice's lightcone, but also to the interior of the lightcone. In other words, quantum information in $(2+1)$-dimensions propagates slower than light via a massless field. This is in agreement with the violations of the strong Huygens principle that occur in $(2+1)$-dimensional Minkowski spacetime.

While we have come to this conclusion just by focusing on the smearing function $F_{\textsc{b}1}(\bm x)$ --- since it is the only one out of the $F_{\textsc{b}i}(\bm x)$ for which we could obtain an integral expression of the form \eqref{eq:FB1_2D} with a closed-form integrand --- let us, for the sake of completeness, now verify numerically that the smearing functions $F_{\textsc{b}2}(\bm x)$ and $F_{\textsc{b}3}(\bm x)$ also have support inside of the light cone, i.e. for $r<\Delta$. Analogous to the (3+1)-dimensional case, let us suppose that Alice's smearing function $F_\textsc{a}(\bm x)$ is given by the Gaussian
\begin{equation}
\label{eq:FA_gaussian_2}
    F_\textsc{a}(\bm x)=\frac{1}{(\sqrt{\pi}\sigma)^2}\exp\left(-\frac{|\bm x|^2}{\sigma^2}\right).
\end{equation}
Then, in Fig.~\ref{fig:Smearings_2D} we indeed find that all three that Bob's smearing functions $F_{\textsc{b}i}(\bm x)$ have support for $r<\Delta$, and hence, as already stated above, we conclude that in order to recover Alice's quantum message in $(2+1)$-dimensional flat spacetime, Bob must couple to the massless field inside of Alice's future light cone. 

\begin{figure}
    \centering
    \includegraphics[width=\linewidth]{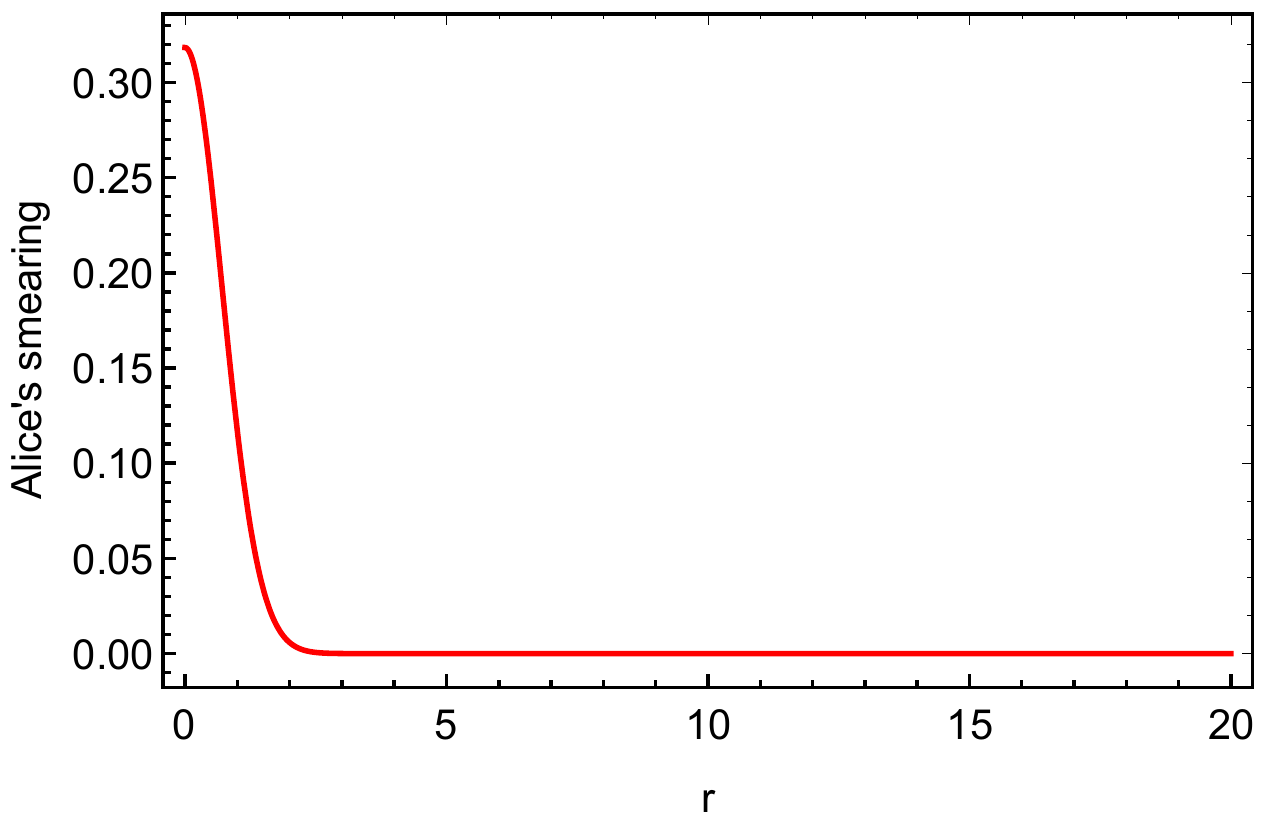}
    \includegraphics[width=\linewidth]{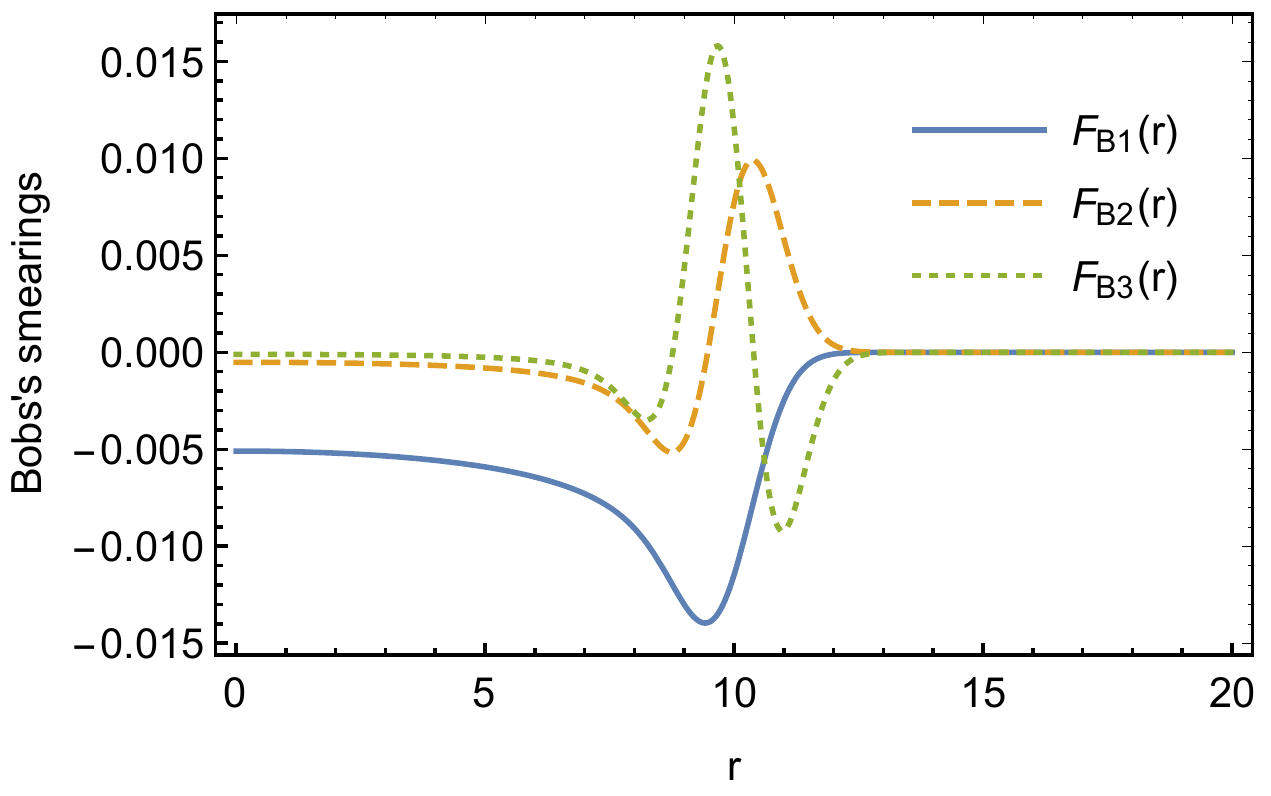}
    \caption[Propagation of quantum information through a massless field in $(2+1)$-dimensional flat spacetime]{$(2+1)$-dimensions. \textbf{Top}: Location at which Alice couples to the massless field at time $\taa$, given by Eq.~\eqref{eq:FA_gaussian_2}, with $\sigma=1$ and where $r=|\bm x|$ is measured in units of $\sigma$. \textbf{Bottom}: Bob's smearing functions, which dictate where in space Bob needs to couple to the field at time $\tbb=\taa+\Delta$ in order to receive Alice's message (we set $\Delta =10$). Note that all three of Bob's smearing functions have support inside of the light cone, i.e. they are only polynomially, rather than exponentially, suppressed for $r\ll\Delta$.}
    \label{fig:Smearings_2D}
\end{figure}

\section{Broadcasting quantum information}
\label{sec:broadcasting_q_info}

In the previous section we have obtained a better understanding of how quantum information propagates through a quantum field by answering the question: Where in space does Bob need to be located if he wants to receive the quantum message that Alice broadcast through the quantum field? Indeed, we found that the answer depends on the spacetime in which Alice and Bob are located. For instance, in $(3+1)$-dimensional Minkowski spacetime, Bob needs to be smeared across Alice's entire light cone, while in the $(2+1)$-dimensional case he also needs to cover the interior of the light cone. In particular, note that in both spacetimes Alice's message is broadcast isotropically in all spatial directions, which is, of course, simply a consequence of the fact that Alice's coupling to the field was fully isotropic.

Let us now consider the relevant case of $(3+1)$-dimensional Minkowski spacetime. Then, although we are fortunate that in this case Bob does not need to cover the interior of Alice's light cone in order to receive her full signal, from a practical perspective it is still very restrictive to require that Bob covers the surface of the lightcone itself (i.e. without the interior), as we found is required in order for him to receive the entirety of Alice's quantum message.

A natural question then arises: Is Alice able to transmit quantum information to Bob if he only covers a part of her light cone? This question is relevant, for instance, if Alice wants to broadcast her information to multiple disjoint receivers, each located in a different spatial direction relative to Alice. In fact, this question was partially answered in Ref.~\cite{Jonsson2018}, where the authors showed that, in a flat spacetime of any dimension, it is not possible for Alice to send any amount of quantum information to multiple \textit{identical} Bobs.\footnote{Of course, from the no-cloning theorem~\cite{Wootters1982} it is clear that Alice cannot perfectly send quantum information to multiple identical Bobs, since this would amount to her quantum state being cloned. The importance of the result in \cite{Jonsson2018} is that it showed this to be true for \textit{any} amount of quantum information, no matter how small.} In this section we will attempt to circumvent this result by considering \textit{non-identical} Bobs.

More concretely, let us consider the setup shown in Fig.~\ref{fig:two_bobs_diagram}, in which two Bobs are trying to recover the message which Alice broadcast into the field. Both Bobs are spherically symmetric, with Bob $B_1$ covering the region of space given by $r<r_0$, and Bob $B_2$ covering the region $r>r_0$. We consider this setup both for its computational simplicity (owing to the fact that spherical symmetry is preserved), as well as the fact that the Bobs in this setup are not identical, thus allowing us to potentially overcome the limitations imposed upon identical Bobs~\cite{Jonsson2018}, as discussed above. Despite the simplicity of the setup however, it will nevertheless provide us with interesting insights into the broadcasting of quantum information through a relativistic quantum field. 

\begin{figure}
    \centering
    \includegraphics[width=\linewidth]{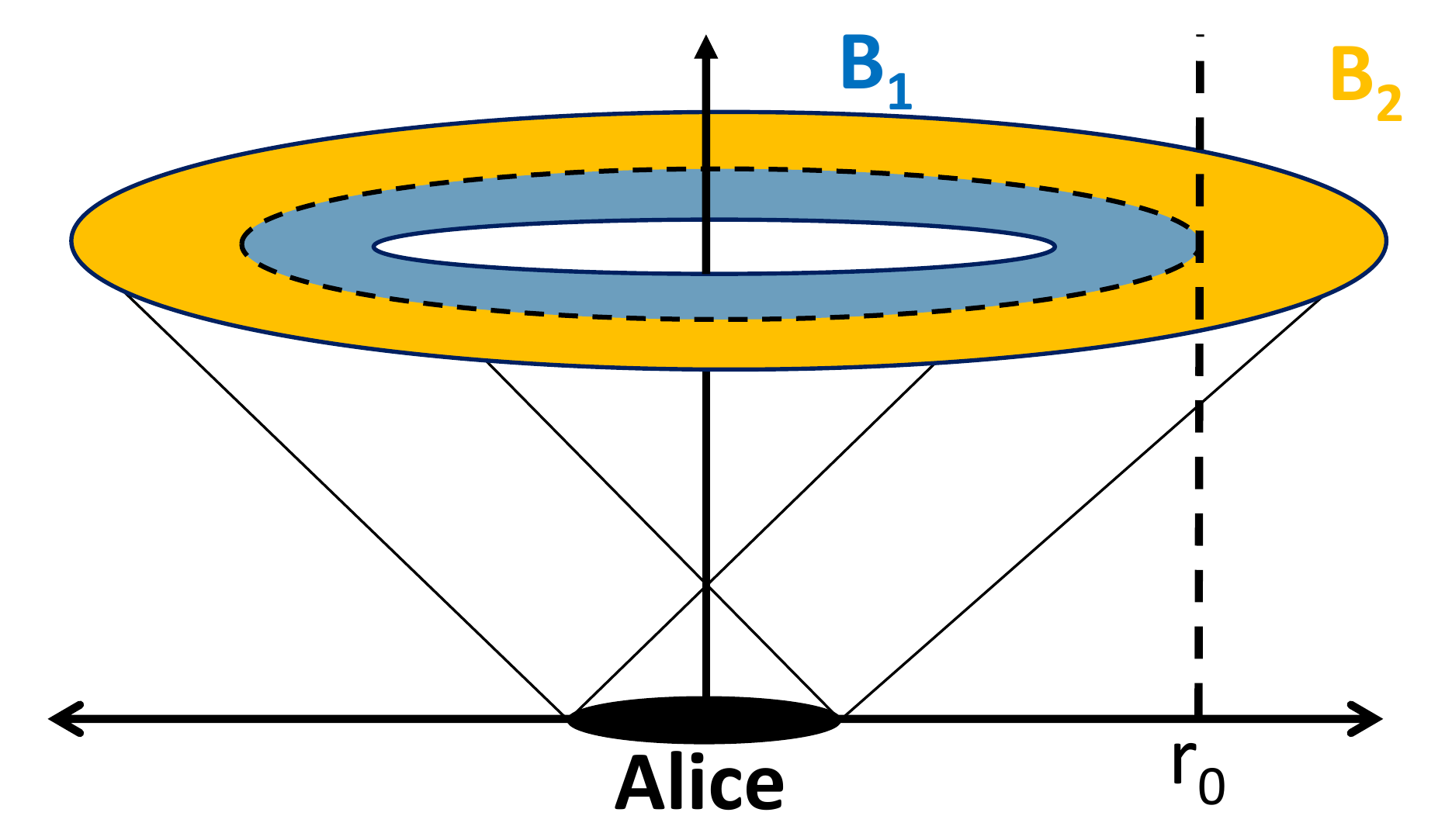}
    \caption[Spacetime diagram of quantum information broadcasting setup from Alice to two Bobs]{$(3+1)$-dimensional quantum information broadcasting setup considered in this section: Alice attempts to send quantum information to two spherically symmetric Bobs, $B_1$ and $B_2$, separated by the radius $r=r_0$.}
    \label{fig:two_bobs_diagram}
\end{figure}

To proceed, let us start by setting the initial state for both Bobs $B_1$ and $B_2$ to be $\ket{+_y}$, and Alice's initial state to be the maximally mixed state, $\hat\rho_{\textsc{a},0}=\frac{1}{2}\mathds 1$. We set the smearing $F_\textsc{a}(\bm x)$ of Alice's detector to be a Gaussian of width $\sigma$, given by Eq.~\eqref{eq:FA_gaussian}. Then, as we saw in Sec.~\ref{sec:testing_quantum_channel}, if Bob wants to recover the entirety of Alice's message, he needs to be able to couple to the field (and its conjugate momentum) via three different smearing functions, $F_{\textsc{b}i}(\bm x)$, given by Eqs.~\eqref{eq:FB1_d}-\eqref{eq:FB3_d}. Namely, 
\begin{align}
\label{eq:FB1_e}
    F_{\textsc{b}1}(\bm x)
    &=
    -
    \int\d[d]{\bm x'}
    F_\textsc{a}(\bm x-\bm x')
    \frac{\delta(|\bm x'|-\Delta)}{4\pi |\bm x'|},\\
    \label{eq:FB2_e}
    F_{\textsc{b}2}(\bm x)
    &=
    -
    \int\d[d]{\bm x'}
    F_\textsc{a}(\bm x-\bm x')
    \frac{\delta'(|\bm x'|-\Delta)}{4\pi |\bm x'|},\\
    \label{eq:FB3_e}
    F_{\textsc{b}3}(\bm x)
    &=
    -
    \int\d[d]{\bm x'}
    F_\textsc{a}(\bm x-\bm x')
    \frac{\delta''(|\bm x'|-\Delta)}{4\pi |\bm x'|}.
\end{align}
This is the ideal case however, where Bob has access to the entirety of Alice's lightcone. We now want to consider the less-than-ideal case of two Bobs, $B_1$ and $B_2$, that only have access to spatial regions $r<r_0$ and $r>r_0$, respectively. Hence, let us set the smearing functions for Bob $B_1$ to be
\begin{align}
\label{eq:FB1_Bob1}
    F_{\textsc{b}1}^{(1)}(\bm x)
    &=
    -
    \int\d[d]{\bm x'}
    F_\textsc{a}(\bm x-\bm x')
    \frac{\delta(|\bm x'|-\Delta)}{4\pi |\bm x'|}
    \Theta(r_0-|\bm x'|),\\
    \label{eq:FB2_Bob1}
    F_{\textsc{b}2}^{(1)}(\bm x)
    &=
    -
    \int\d[d]{\bm x'}
    F_\textsc{a}(\bm x-\bm x')
    \frac{\delta'(|\bm x'|-\Delta)}{4\pi |\bm x'|}
    \Theta(r_0-|\bm x'|),\\
    \label{eq:FB3_Bob1}
    F_{\textsc{b}3}^{(1)}(\bm x)
    &=
    -
    \int\d[d]{\bm x'}
    F_\textsc{a}(\bm x-\bm x')
    \frac{\delta''(|\bm x'|-\Delta)}{4\pi |\bm x'|}
    \Theta(r_0-|\bm x'|),
\end{align}
where the subscript $(1)$ indicates Bob $B_1$, and the $\Theta$ functions ensure that these smearings are only non-zero in the ball $r<r_0$ centered on Alice. Similarly, for Bob $B_2$ we set the smearings to be
\begin{align}
\label{eq:FB1_Bob2}
    F_{\textsc{b}1}^{(2)}(\bm x)
    &=
    -
    \int\d[d]{\bm x'}
    F_\textsc{a}(\bm x-\bm x')
    \frac{\delta(|\bm x'|-\Delta)}{4\pi |\bm x'|}
    \Theta(|\bm x'|-r_0),\\
    \label{eq:FB2_Bob2}
    F_{\textsc{b}2}^{(2)}(\bm x)
    &=
    -
    \int\d[d]{\bm x'}
    F_\textsc{a}(\bm x-\bm x')
    \frac{\delta'(|\bm x'|-\Delta)}{4\pi |\bm x'|}
    \Theta(|\bm x'|-r_0),\\
    \label{eq:FB3_Bob2}
    F_{\textsc{b}3}^{(2)}(\bm x)
    &=
    -
    \int\d[d]{\bm x'}
    F_\textsc{a}(\bm x-\bm x')
    \frac{\delta''(|\bm x'|-\Delta)}{4\pi |\bm x'|}
    \Theta(|\bm x'|-r_0),
\end{align}
which only have support in the spatial region $r>r_0$. Additionally, we will keep the condition $\lambda_\phi\lambda_\pi/\sqrt{(2\pi)^3}\sigma^3=\pi/4$ relating the coupling constants $\lambda_\phi$ and $\lambda_\pi$ to the size of the detector $\sigma$, which we recall was necessary in order for Alice to be able to perfectly transmit her quantum message to (a single) Bob.

Having specified the initial quantum states as well as the smearing functions of Alice and both Bobs, we can now proceed to numerically compute the density matrices $\hat\rho_{\textsc{cb}_1}$ and $\hat\rho_{\textsc{cb}_2}$ associated with each Bob, as given by Eq.~\eqref{eq:rho_cb_2}. Then, via Eq.~\eqref{eq:coherent_info}, we can compute the coherent information associated with the channel from Alice to Bob $B_1$, and similarly for Bob $B_2$. The results are shown in Fig.~\ref{fig:C_vs_r0} for two choices of parameters: $\lambda_\phi/\sigma=10$ and $\lambda_\phi/\sigma=1000$.

\begin{figure}
    \centering
    \includegraphics[width=\linewidth]{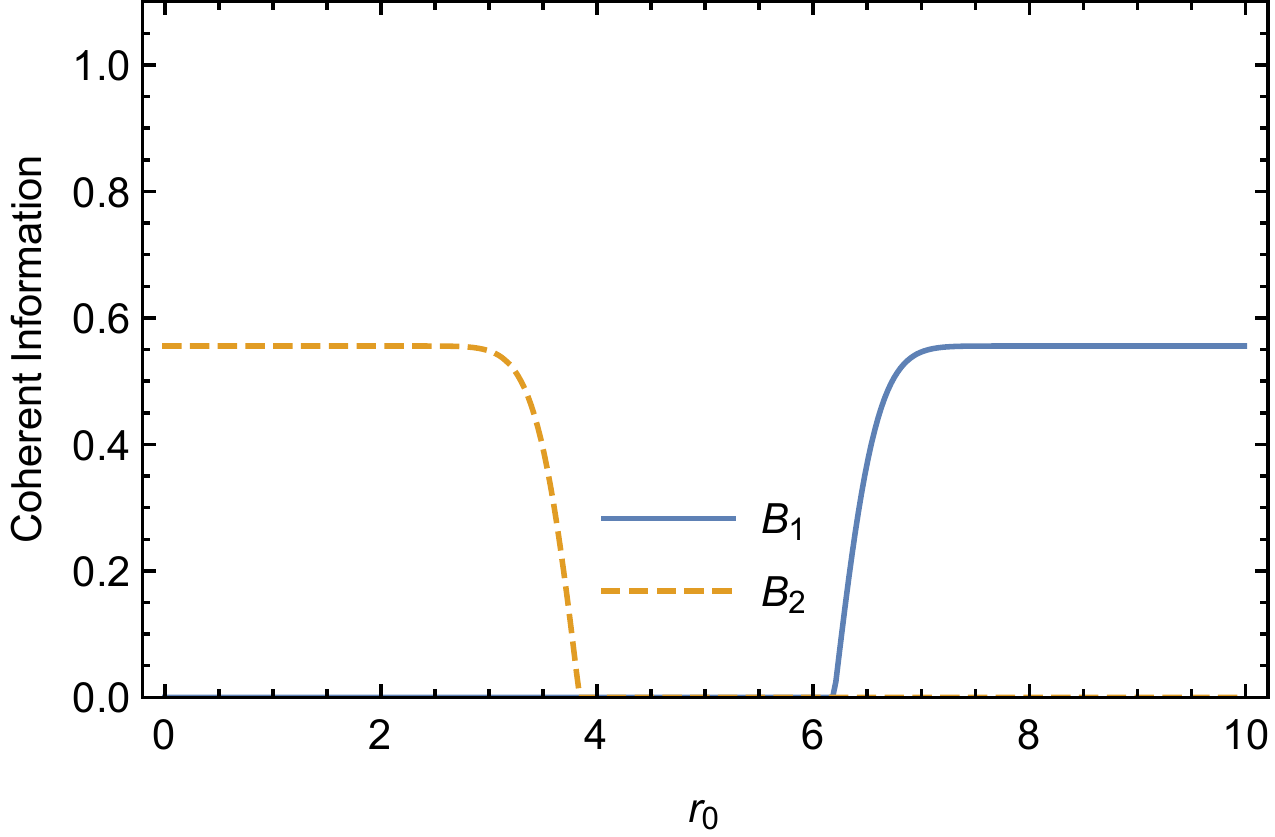}
    \includegraphics[width=\linewidth]{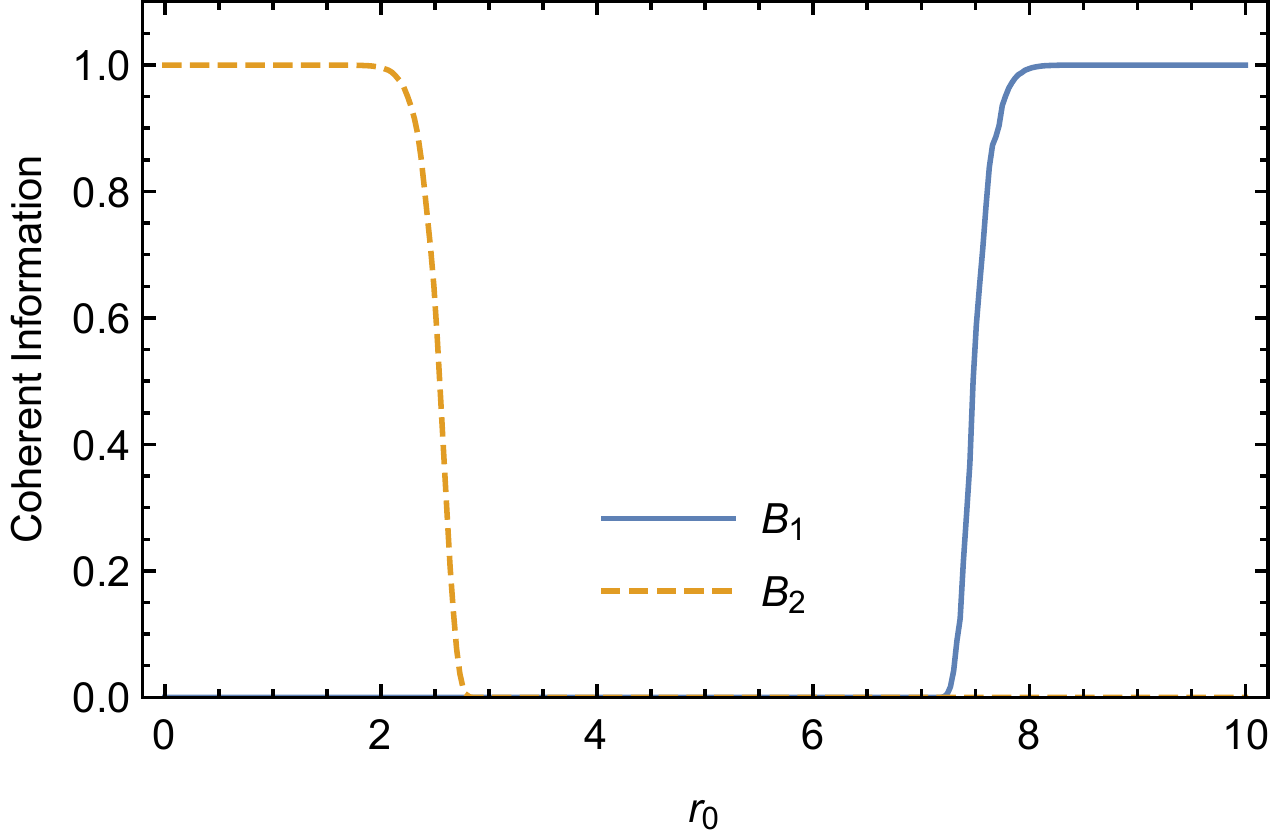}
    \caption{Coherent information versus $r_0$ (the radial separation between Bob $B_1$ and $B_2$) with $\lambda_\phi=10$ (\textbf{top}) and $\lambda_\phi=1000$ (\textbf{bottom}). We set $\sigma=1$ for both plots. Notice that for both choices of parameters, and for any choice of $r_0$, it is not possible for Alice to simultaneously send coherent information to both Bobs.}
    \label{fig:C_vs_r0}
\end{figure}

There are a few interesting points to note regarding Fig.~\ref{fig:C_vs_r0}. First, notice that for small enough $r_0$, Alice can send quantum information to Bob $B_2$ (but not $B_1$). This makes sense, since, as can be seen in Fig.~\ref{fig:two_bobs_diagram}, a small enough value of $r_0$ means that Bob $B_2$ has access to the entire lightcone of Alice, and thus he can recover the full quantum message (which, in $(3+1)$-dimensions propagates on the lightcone). Similarly, for large enough $r_0$ Alice can send quantum information to Bob $B_1$, but not $B_2$.

However, for either of the parameter ratios $\lambda_\phi/\sigma$, it is not possible for Alice to simultaneously broadcast her quantum message to both Bobs, regardless of the value we take for the radius $r_0$ which defines the separation of $B_1$ and $B_2$. In fact we numerically verified that there is no possible value of ratio $\lambda_\phi/\sigma$ which allows Alice to simultaneously broadcast coherent information to both Bobs. This therefore extends the no-quantum-broadcasting result proven in Ref.~\cite{Jonsson2018} for identical detectors to the case of spherically symmetric, non-identical detectors, and it therefore gives supporting evidence to the conjecture that it is not possible to send quantum information through a quantum field to multiple disjoint detectors, identical or not.

Another interesting feature to note in Fig.~\ref{fig:C_vs_r0} is the effect that increasing the ratio $\lambda_\phi/\sigma$ from 10 to 1000 has on the coherent information $I_c(\hat\rho_{\textsc{a},0},\Xi)$ that Alice can transmit to the two Bobs. Namely, we see that for the smaller value of $\lambda_\phi/\sigma$ Alice can transmit coherent information to both Bobs for a larger range of values of $r_0$, but for either Bob $I_c(\hat\rho_{\textsc{a},0},\Xi)$ never exceeds $0.6$. On the other hand, if the coupling strength $\lambda_\phi$ is increased relative to $\sigma$, then there is a smaller range of $r_0$ values for which either Bob can receive coherent information, but in the best case scenario (large $r_0$ in the case of Bob $B_1$ and small $r_0$ in the case of Bob $B_2$) the Bobs can receive the maximum value of coherent information, $I_c(\hat\rho_{\textsc{a},0},\Xi)=1$. In other words, there is a ``rich-get-richer, poor-get-poorer" type of trade-off associated with increasing the ratio $\lambda_\phi/\sigma$: large Bobs will be able to receive more quantum information, at the expense of smaller Bobs not being able to receive any.

We can understand this trade-off on physical grounds, as follows. First of all, we know from our discussion in Sec.~\ref{sec:perfect_quantum_channel} that in order for Alice to perfectly transmit her quantum information to a single Bob the strong coupling condition \eqref{eq:strong_coupling_condition} must be satisfied, which in $(3+1)$-dimensions is given by Eq.~\eqref{eq:strong_coupling_condition_2}: $\lambda_\phi\gg\sigma$. Hence it is not surprising that if the ratio $\lambda_\phi/\sigma$ is increased, then a large Bob $B_1$ or $B_2$ --- who would approximate the single, ideal Bob considered in the previous sections --- would be able to receive more coherent information from Alice. 

Furthermore, it also makes intuitive sense that a larger coupling $\lambda_\phi$ would make it more difficult for smaller, less-than-ideal Bobs to receive quantum information from Alice. 

To understand this, first recall from Sec.~\ref{sec:perfect_quantum_channel} the physics of our perfect quantum channel from Alice to Bob. The first step to the quantum channel consists of Alice encoding her qubit into coherent states of the field, which, for larger values of $\lambda_\phi$ are increasingly more and more orthogonal to one another. Then, Bob attempts to recover the message by performing the DECODE gate between his qubit and the field, as shown schematically in Fig.~\ref{fig:Xi}. The DECODE gate, defined as the inverse to the ENCODE gate, first entangles Bob's qubit with the coherent field states, and then attempts to disentangle the field so that Alice's qubit state is coherently transmitted to Bob. However, in order for this final disentangling step to be performed successfully, Bob must have access to the entire quantum message sent out by Alice --- i.e. Bob must have access to the entirety of Alice's lightcone in $(3+1)$-dimensional flat spacetime.

However, in this section we are manifestly considering the scenario where a less-than-ideal Bob ($B_1$ or $B_2$) does \textit{not} have access to the entirety of Alice's lightcone, and hence in his decoding process he will \textit{not} be able to completely disentangle his qubit from the field. Hence, following the decoding procedure the field carries partial knowledge of Bob's state, i.e. Alice's state, which she hoped to transmit to Bob. In other words, a portion of Alice's message will remain in the field, and hence, by a no-cloning type of intuition, the full message cannot get transmitted to this Bob. Now recall that the amount of overlap between the coherent field components appearing in the entangled state \eqref{entangled} determines how much Bob remains entangled with the field. This amount of overlap is in turn determined by the coupling intensity. Hence, a larger value of the coupling $\lambda_\phi$, which ensures greater orthogonality between the coherent field states, requires Bob to cover a larger portion of Alice's light cone in order to receive her quantum message.

As a final note, it can be argued that some optical quantum channels are implemented not from a fundamental light-matter isotropic coupling, but rather with highly directional light sources. Indeed, our results show that if Alice instead wanted to send her quantum message to a single Bob localized in some specified solid angle $\Omega<4\pi$ relative to her, it would be much more prudent for her to change the way in which she couples to the field, so that it is not isotropic, but rather so that she only couples to those field modes with wavevectors pointing in Bob's direction (e.g., using mirrors to collimate the signal). In this way the quantum information that Alice encodes in the field would only travel towards Bob and not in all directions, and Bob would be able to recover more of the quantum message, rather than only a small fraction. One would yet expect timelike leakage of the information in this case in spacetimes where strong Huygens is violated, but we leave the study of such non-isotropic couplings for a future work.

\section{Summary and Conclusions}
\label{sec:conclusions}

We studied how a relativistic quantum field can be used to transmit quantum information between spacetime observers Alice and Bob, who couple to the field via particle detectors. 

When it comes to quantum communication through a field (e.g., the electromagnetic field) one may expect that quantum information can be transferred from an atom to the field and from the field to another atom as long as the transmission is light-like. However, we showed that this naive assumption is not true, and that the question of `where does the quantum information go' in time and space, when it is encoded into a quantum field from an atomic system, is non-trivial.

Namely, we have shown that in spacetime, the quantum information that is originally contained in a particle detector (e.g., a qubit encoded in an atomic system) is transmitted following a linear coupling between the detector and a quantum field. We also have shown that the quantum information localization and propagation through the field via the usual light-matter coupling displays unintuitive features. 
Furthermore, we have quantified ---by means of the quantum channel capacity between Alice and Bob--- how the localization and spatial profile of Bob affects his ability to recover the quantum information left in the field by Alice.

Concretely, we began by constructing a perfect field-mediated quantum channel from Alice to Bob, i.e., one for which the quantum channel capacity is the theoretically maximum value. The channel can be implemented by Alice first coupling to the field via a local unitary $\Ua$, which serves to encode Alice's qubit into the field, followed by Bob coupling to the field via a local unitary $\Ub$, which decodes the qubit from the field and onto Bob's detector. 

The unitaries $\Ua$ and $\Ub$ defining our quantum channel are each generated by interaction Hamiltonians that couple Alice and Bob's detectors to the field only at discrete instants in time, and hence allow for a non-perturbative approach to the problem of time-evolution. Indeed, such a non-perturbative approach is necessary, since, as we showed, the field-mediated quantum channel from Alice to Bob can only be a perfect quantum channel if the observers are strongly (i.e. non-perturbatively) coupled to the field. 

In particular, the unitaries $\Ua$ and $\Ub$ in our construction, each take the form of a product of \textit{two} simple-generated (i.e. rank-1 generated) unitaries. We showed that these are the simplest possible unitaries leading to a quantum channel with a non-zero quantum capacity. That is, if either $\Ua$ or $\Ub$ consists of a \textit{single} rank-1 unitary, then the channel from Alice to Bob necessarily has zero quantum capacity. In this sense, the channel which we construct is the simplest possible field-mediated quantum channel from Alice to Bob with a non-zero quantum capacity.

Following our mathematical construction of the simplest possible perfect quantum channel (which upper-bounds the performance of any other possible channel), we used it to better understand how quantum information propagates through a relativistic quantum field. In particular, we asked the following question: If Alice encodes a quantum message into a quantum field at time $t_\textsc{a}$ by coupling to the field in a spatial region characterized by the smearing function $F_\textsc{a}(\bm x)$, then where in space does Bob have to be located at time $t_\textsc{b}>t_\textsc{a}$ in order to fully receive Alice's message?

We answered this question by showing that if Bob wants to guarantee that he fully receives Alice's quantum message, then he must have access to the region of space containing the supports of a set of smearing functions $F_{\textsc{b}i}(\bm x)$, for $i\in\{1,2,3\}$. These smearing functions are defined in terms of Alice's smearing $F_\textsc{a}(\bm x)$ and the time difference $\Delta=t_\textsc{b}-t_\textsc{a}$, and they completely characterize the flow of quantum information through a Klein-Gordon field of arbitrary mass $m$ in a flat spacetime of arbitrary dimension.

To better understand this general result, we then considered the particular cases of quantum information propagation through massless fields in $(2+1)$- and $(3+1)$-dimensional flat spacetimes. In $(3+1)$-dimensions we found that Bob can fully recover Alice's quantum message if he has access to her future light cone, which allowed us to conclude that in this spacetime quantum information propagates at the speed of light through the massless field. On the other hand, in the $(2+1)$-dimensional case we found that Bob additionally must have access to the full interior of Alice's lightcone in order to recover the entire message. Hence, in $(2+1)$-dimensional flat spacetime quantum information propagates subluminally through a massless field, despite the fact that the field quanta travel at the speed of light.

While this latter result may at first seem surprising, it can be understood by studying the validity of the strong Huygens principle. Indeed, as is well known, the strong Huygens principle does not hold in most spacetimes --- including even dimensional Minkowski spaces --- and in principle, information can propagate slower than light in these spacetimes~\cite{McLenaghan1974}. While this has previously been extensively investigated for \textit{classical} information transmission~\cite{Jonsson2015,Blasco2015,Blasco2016,Simidzija2017}, our work presented here is, to our knowledge, the first study of the effects of strong Huygens violations on \textit{quantum} information transmission.

Having understood where in space an ideal Bob needs to be located in order to perfectly receive the quantum message that Alice sends through the field, we considered the less-than-ideal situation where Bob only covers a part of the spacetime region in which Alice's message lives. This situation is interesting from the perspective of quantum information broadcasting, a setup in which Alice hopes to simultaneously transmit at least a part of her quantum message to multiple disjoint Bobs. While the no-cloning theorem~\cite{Wootters1982} precludes a perfect transmission of quantum information to multiple receivers, there appears, a priori, no reason to suspect that at least a small amount of quantum information could not be recovered by each of the Bobs.

However, as was shown in Ref.~\cite{Jonsson2018}, it is in fact impossible for Alice to broadcast \textit{any} amount of quantum information to multiple \textit{identical} Bobs, a result that was proven for any spacetime dimension by noting that the quantum channel from Alice to any such Bob is what is called anti-degradable~\cite{Holevo2008}. Nevertheless this still leaves open the possibility for broadcasting quantum information to multiple, \textit{non-identical}, disjoint Bobs, which we proceeded to study.

More concretely, we considered the case of two spherically symmetric Bobs, $B_1$ and $B_2$, covering the regions of $(3+1)$-dimensional space given by $|\bm x|<r_0$ and $|\bm x|\ge r_0$ (with $r_0$ some fixed radius), attempting to recover the quantum information sent out via a massless quantum field by an emitter Alice located at $\bm x=0$. (The setup is depicted in Fig.~\ref{fig:two_bobs_diagram}.) We found that, regardless of the choice of setup parameters --- such as the separation radius $r_0$ and the field coupling strength $\lambda_\phi$ of the detectors to the quantum field --- it is not possible for Alice to simultaneously broadcast a non-zero amount of coherent information (a lower bound on the quantum channel capacity) to both Bobs. This gives support to the conjecture that it is not possible for Alice to broadcast quantum information to multiple disjoint Bobs, identical or not. 

It should be emphasized that the strong localization requirements imposed on Bob who wants to recover the entirety of Alice's message are strongly related  to the assumption that the coupling of Alice to the field is isotropic. In other words, if Alice couples isotropically to the field, it is fully expected that her quantum message will propagate symmetrically in all spatial directions away from the point of emission, and from the prior studies of classical communication through quantum fields~\cite{Jonsson2015,Blasco2015,Jonsson2016,Simidzija2017}, it is perhaps not even surprising that Alice's quantum message also propagates in the interior of her light cone in strong Huygens violating spacetimes. However, there is a crucial finding that the present work brings to light which is a vital distinction between the propagation of classical and quantum information through a quantum field. Unlike in the classical case, a receiver Bob who wants to receive Alice's entire quantum message \textit{must} couple to the field in the entire spacetime region in which the message is located. Hence, while a classical information receiving Bob will find it beneficial that Alice's message is so delocalized in space, since he can fully recover its content by coupling to \textit{any} region of spacetime containing the message (even, perhaps, a region in timelike separation from Alice), a quantum information receiving Bob will find this same feature to be a hindrance, since the recovery of the the quantum message requires him to have access to the \textit{entire} region of spacetime containing the message. Of course, Alice could ammeliorate this issue for the latter Bob by imposing a directionality in her coupling to the field, so that the her message propagates preferentially towards the intended receiver, but she would still not be able to prevent the leakage of her message inside the timelike area of her future light cone, which is solely a consequence of strong Huygens violations in certain spacetimes and is out of Alice's control.

Our study of quantum information broadcasting also led to an interesting result relating the coupling strength $\lambda_\phi$ of Alice and Bob's detectors to the quantum field and the minimum size that a given Bob must be in order to receive at least a part of Alice's quantum message. Namely, we showed that there is a ``rich-get-richer, poor-get-poorer" type of trade-off associated with increasing the coupling strength $\lambda_\phi$, whereby very large Bobs are able to receive more quantum information from Alice, at the cost of smaller Bobs not being able to receive any. 

Physically, this trade-off arises due to the fact that an increased coupling $\lambda_\phi$ ensures that Alice's qubit is stored more coherently in the field, and hence a receiver Bob who has access to the entire portion of the field containing the qubit can better recover the qubit using his own detector. The downside of such a highly coherent encoding of Alice's qubit into the field however, is that if Bob is not able to fully access the region of the field containing the qubit, then a significant portion of Alice's message will remain in the field after Bob attempts to recover it. And since the no-cloning theorem makes it impossible for Alice's state to be simultaneously encoded in both the field and Bob's detector, we can thus understand intuitively why a spatially limited Bob would struggle to receive quantum information from Alice if $\lambda_\phi$ is large.

Finally, let us mention that it should be very interesting to generalize the present results, for example, to the case of multiple senders and multiple receivers. There, generalizing the study in \cite{QSWAhmadzadegan} from classical to quantum information, it should be possible to show that not only the classical channel capacity but also the quantum channel capacity can be modulated and enhanced by beam shaping by suitably pre-entangling the emitters.

\acknowledgements

P.S. acknowledges the support of the NSERC CGS-M Scholarship and Ontario Graduate Scholarship. A.K. and E.M.M. acknowledge support through the NSERC Discovery Program.  E.M.M. also acknowledges funding through his Ontario Early Researcher Award. A.K. acknowledges support through a Google Faculty Research Award.

\onecolumngrid
\appendix

\section{Some technical quantum information results}
\label{Appendix:quantum_info_results}

Here we some more technical quantum information theory results, which we make use of throughout the main text.

\begin{mydef}
Let $\hat\rho_\textsc{cb}$ be a state on $\mathcal H_\textsc{c}\otimes\mathcal H_\textsc{b}$. The \emph{conditional quantum entropy} $S(C|B)_{\hat\rho_\textsc{cb}}$ is defined as
\begin{equation}
\label{eq:conditional_entropy}
    S(C|B)_{\hat\rho_\textsc{cb}}
    :=
    S(\hat\rho_\textsc{cb})
    -
    S(\hat\rho_\textsc{b}),
\end{equation}
where $S(\cdot)$ denotes the von Neumann entropy and $\hat\rho_\textsc{b}:=\tr_\textsc{c}\hat\rho_\textsc{cb}$.
\end{mydef}
Note that, with this definition, the coherent information $I_c(\hat\rho_{\textsc{a},0},\Xi)$ of a quantum channel $\Xi$ from A to B and the input state $\hat\rho_{\textsc{a},0}$, as defined by Eq.~\eqref{eq:coherent_info}, can be written as
\begin{equation}
    \label{eq:coherent_info_2}
    I_c(\hat\rho_{\textsc{a},0},\Xi)
    =
    -S(C|B)_{\hat\rho_\textsc{cb}},
\end{equation}
where, recall, $\hat\rho_\textsc{cb}$ is the output of the channel $\mathds 1_\textsc{c}\otimes\Xi$ acting on a purification of $\hat\rho_{\textsc{a},0}$.

\begin{lemma}
\label{lemma:concavity}
The function taking the input $\hat\rho_\textsc{cb}$ and producing the output $S(C|B)_{\hat\rho_\textsc{cb}}$ is a concave function, i.e. 
\begin{equation}
    S(C|B)_{\lambda\hat\rho_1+(1-\lambda)\hat\rho_2}
    \ge
    \lambda S(C|B)_{\hat\rho_1}
    +
    (1-\lambda)S(C|B)_{\hat\rho_2},
\end{equation}
for any $0\le\lambda\le 1$ and states $\hat\rho_1$ and $\hat\rho_2$ on $\mathcal H_\textsc{c}\otimes\mathcal H_\textsc{b}$.
\end{lemma}

\begin{proof}
The proof presented here is inspired by the sketch of the proof in~\cite{Watrous2018}. We start by considering the state $\hat\rho_\textsc{cbe}$ on $\mathcal H_\textsc{c}\otimes\mathcal H_\textsc{b}\otimes\mathcal H_\textsc{e}$ defined by
\begin{equation}
    \hat\rho_\textsc{cbe}
    :=
    \lambda\hat\rho_1\otimes\ket0\bra0+(1-\lambda)\hat\rho_2\otimes\ket0\bra0,
\end{equation}
where $\{\ket0,\ket1\}$ forms an orthonormal basis of the auxiliary qubit space $\mathcal H_\textsc{e}$. Then, the strong subadditivity of the von Neumann entropy reads~\cite{Lieb1973}
\begin{equation}
    S(\hat\rho_\textsc{cbe})
    +
    S(\hat\rho_\textsc{b})
    \le
    S(\hat\rho_\textsc{cb})
    +
    S(\hat\rho_\textsc{be}),
\end{equation}
where $\hat\rho_\textsc{b}:=\tr_\textsc{ce}\hat\rho_\textsc{cbe}$, $\hat\rho_\textsc{cb}:=\tr_\textsc{e}\hat\rho_\textsc{cbe}$ and $\hat\rho_\textsc{be}:=\tr_\textsc{c}\hat\rho_\textsc{cbe}$. Then, noting that $\hat\rho_\textsc{cb}=\lambda\hat\rho_1+(1-\lambda)\hat\rho_2$ and making use of the definition \eqref{eq:conditional_entropy} for $S(C|B)_{\hat\rho_\textsc{cb}}$ we find
\begin{equation}
\label{eq:conditional_entropy_bound}
    S(C|B)_{\lambda\hat\rho_1+(1-\lambda)\hat\rho_2}
    \ge 
    S(\hat\rho_\textsc{cbe})
    -
    S(\hat\rho_\textsc{be}).
\end{equation}

Let us now evaluate $S(\hat\rho_\textsc{cbe})$. We obtain
\begin{align}
    S(\hat\rho_\textsc{cbe})
    &:=
    -\tr \hat\rho_\textsc{cbe}\log_2 \hat\rho_\textsc{cbe}\notag\\
    &=
    -\tr \hat\rho_\textsc{cbe}
    \left(
    \log_2(\lambda\hat\rho_1)\otimes\ket0\bra0
    +
    \log_2((1-\lambda)\hat\rho_2)\otimes\ket1\bra1
    \right)\notag\\
    &=
    -\tr\left(
    \lambda\hat\rho_1\log_2(\lambda\hat\rho_1)
    \otimes\ket0\bra0+
    (1-\lambda)\hat\rho_2\log_2((1-\lambda)\hat\rho_2\otimes\ket1\bra1
    \right)\notag\\
    &=
    -\tr\lambda\hat\rho_1\log_2(\lambda\hat\rho_1)
    -\tr(1-\lambda)\hat\rho_2\log_2((1-\lambda)\hat\rho_2)\notag\\
    &=
    S(\lambda\hat\rho_1)+S((1-\lambda)\hat\rho_2).
    \label{eq:S(rho_cbe)}
\end{align}
By an analogous calculation we find
\begin{equation}
    \label{eq:S(rho_be}
    S(\hat\rho_\textsc{be})
    =
    S(\lambda\tr_\textsc{c}\hat\rho_1)
    +
    S((1-\lambda)\tr_\textsc{c}\hat\rho_2).
\end{equation}
Then, combining Eqs.~\eqref{eq:conditional_entropy_bound}-\eqref{eq:S(rho_be} we obtain
\begin{align}
    S(C|B)_{\lambda\hat\rho_1+(1-\lambda)\hat\rho_2}
    &\ge 
    S(\lambda\hat\rho_1)+S((1-\lambda)\hat\rho_2)-S(\lambda\tr_\textsc{c}\hat\rho_1)
    -
    S((1-\lambda)\tr_\textsc{c}\hat\rho_2).
\end{align}
Using the identity $S(\lambda\hat\rho)=\lambda\log_2\lambda+\lambda S(\hat\rho)$, which is straightforwardly proven by working in the eigenbasis of $\hat\rho$, the above expression simplifies to
\begin{align}
    S(C|B)_{\lambda\hat\rho_1+(1-\lambda)\hat\rho_2}
    &\ge 
    \lambda
    \left[S(\hat\rho_1)-S(\tr_\textsc{c}\hat\rho_1)\right]
    +
    (1-\lambda)
    \left[S(\hat\rho_2)-S(\tr_\textsc{c}\hat\rho_2)\right].
\end{align}
Finally, using the definition \eqref{eq:conditional_entropy} for the conditional entropy $S(C|B)_{\hat\rho}$ we find
\begin{align}
    S(C|B)_{\lambda\hat\rho_1+(1-\lambda)\hat\rho_2}
    \ge
    \lambda S(C|B)_{\hat\rho_1}
    +
    (1-\lambda)S(C|B)_{\hat\rho_2},
\end{align}
which completes the proof.
\end{proof}
We can now prove a useful result regarding the coherent information $I_c(\hat\rho_{\textsc{a},0},\Xi)$.
\begin{lemma}
Let $\Xi$ be a quantum channel from states on $\mathcal H_\textsc{a}$ to states on $\mathcal H_\textsc{b}$, let $\hat\rho_{\textsc{a},0}$ be a state on $\mathcal H_\textsc{a}$, and let $\hat\rho_\textsc{cb}$ be the output of the channel $\mathds 1_\textsc{c}\otimes\Xi$ applied on the purification of $\hat\rho_{\textsc{a},0}$. Then, $I_c(\hat\rho_{\textsc{a},0},\Xi)\le 0$ if $\hat\rho_\textsc{cb}$ is separable.
\end{lemma}

\begin{proof}
Assume $\hat\rho_\textsc{cb}$ is separable. Then, it is possible to find pure states $\ket{b_i}\in\mathcal H_\textsc{b}$ and $\ket{c_i}\in\mathcal H_\textsc{c}$ along with real numbers $p_i>0$ such that
\begin{equation}
    \hat\rho_\textsc{cb}
    =
    \sum_i p_i
    \ket{c_i}\bra{c_i}\otimes
    \ket{b_i}\bra{b_i}.
\end{equation}
From Eq.~\eqref{eq:coherent_info_2} we have $I_c(\hat\rho_{\textsc{a},0},\Xi)=-S(C|B)_{\hat\rho_\textsc{cb}}$ and hence from Lemma~\ref{lemma:concavity} we find
\begin{align}
\label{eq:coherent_info_bound}
    -I_c(\hat\rho_{\textsc{a},0},\Xi)
    \ge 
    \sum_i p_i
    S(C|B)_{\ket{c_i b_i}\bra{c_i b_i}},
\end{align}
where $\ket{c_i b_i}:=\ket{c_i}\otimes\ket{b_i}$ are pure, separable states on $\mathcal H_\textsc{c}\otimes\mathcal H_\textsc{b}$. Since $S(\ket{c_i b_i})=S(\ket{b_i})=0$ we see from Eq.~\eqref{eq:conditional_entropy} that $S(C|B)_{\ket{c_i b_i}\bra{c_i b_i}}=0$, and hence Eq.~\eqref{eq:coherent_info_bound} reads
\begin{align}
    -I_c(\hat\rho_{\textsc{a},0},\Xi)
    \ge
    0,
\end{align}
which completes the proof.
\end{proof}

\twocolumngrid
\bibliography{references}
\bibliographystyle{apsrev4-1}


%

\end{document}